\documentclass[11pt]{amsart} 

\usepackage{amsmath,amsthm}
\usepackage{amssymb}
\usepackage{booktabs}
\usepackage[T1]{fontenc}
\usepackage[utf8]{inputenc}
\usepackage{lmodern}
\usepackage{amsmath,amssymb,amsthm,mathtools}
\usepackage{enumitem}
\usepackage{microtype}
\usepackage[hidelinks]{hyperref}
\usepackage[margin=1in]{geometry}

\usepackage{enumitem}

\usepackage{graphicx}

\usepackage{xcolor,framed,tikz,float,hyperref,url}

\definecolor{shadecolor}{gray}{0.9}
\usetikzlibrary{arrows.meta,calc}

\usepackage{comment}

\newtheorem{theorem}{Theorem}
\newtheorem{corollary}[theorem]{Corollary}
\newtheorem{lemma}[theorem]{Lemma}
\newtheorem{claim}[theorem]{Claim}
\newtheorem{proposition}[theorem]{Proposition}

\newtheorem{observation}[theorem]{Observation}

\theoremstyle{definition}
\newtheorem{definition}[theorem]{Definition}
\newtheorem{remark}[theorem]{Remark}

\newcommand{\Pclass}{\mathsf{P}}
\newcommand{\DTIME}{\mathsf{DTIME}}
\newcommand{\NP}{\mathsf{NP}}
\newcommand{\coNP}{\mathsf{coNP}}
\newcommand{\PH}{\mathsf{PH}}
\newcommand{\SigmaP}[2]{\mathsf{\Sigma}_{#1}^{{\mathsf P},{#2}}}
\newcommand{\PiP}[2]{\mathsf{\Pi}_{#1}^{{\mathsf P},{#2}}}
\newcommand{\DeltaP}[2]{\mathsf{\Delta}_{#1}^{{\mathsf P},{#2}}}
\newcommand{\Prb}{\mathbf{Pr}}

\newcommand{\one}{\mathbf{1}}
\newcommand{\bits}{\{0,1\}}

\newcommand{\poly}{\operatorname{poly}}
\newcommand{\Sipser}{\operatorname{Sipser}}
\newcommand{\PP}{\mathsf{PP}}
\newcommand{\PSPACE}{\mathsf{PSPACE}}
\newcommand{\ParityP}{\mathsf{\oplus P}}
\newcommand{\UP}{\mathrm{UP}}
\newcommand{\coUP}{\mathrm{coUP}}
\newcommand{\RP}{\mathrm{RP}}
\newcommand{\NPclass}{\mathsf{NP}}
\newcommand{\IP}{\mathsf{IP}}

\newcommand{\val}{\operatorname{val}}
\newcommand{\Prob}{\mathbf{Pr}}

\newcommand{\ii}{\mathrm{i.o.}}

\providecommand{\Pclass}{\mathsf{P}}
\providecommand{\NP}{\mathsf{NP}}
\providecommand{\coNP}{\mathsf{coNP}}
\providecommand{\DP}{\mathsf{DP}}
\providecommand{\BH}{\mathsf{BH}}

\providecommand{\bits}{\{0,1\}}
\providecommand{\Prb}{\mathbf P}

\begin{document}

\title[Separations with Immunity]{Separations with Immunity Relative to a Random Oracle in Computational Complexity}

\author{Gabriel Istrate}
\address{
Faculty of Mathematics and Computer Science\\ 
University of Bucharest\\
Bucharest, Romania}
\email{gabriel.istrate@unibuc.ro} 

\date{}


\begin{abstract}
Rossman, Servedio, and Tan \cite{1} proved that the polynomial hierarchy is infinite
relative to a random oracle.  This paper strengthens this celebrated result by considering \emph{strong 
separations}, witnessed by immune languages.  

The main result of the paper shows that with respect to a random oracle the levels of the polynomial hierarchy separate with immunity.  Specifically, with probability one over a random oracle $A$, for every $k\geq 1$
there is a language in $\SigmaP{k}{A}$ that is immune to $\PiP{k}{A}$, and,
symmetrically, a language in $\PiP{k}{A}$ that is immune to $\SigmaP{k}{A}$. We thus extend classical results due to Bennett and Gill \cite{2} and Vereshchagin \cite{3,63}. 

We accomplish this by developing a "rare-or-wrong" approach to strong separations.  We highlight the flexibility of the approach by strengthening other separations with respect to a random oracle from the complexity-theoretic literature.  \end{abstract}

\maketitle

\tableofcontents 

\section{Introduction}
\label{sec:introduction}

Complexity theory is (or at least promised to be\footnote{\emph{How well} Complexity Theory has fulfilled this promise is another discussion; one that we prefer, however, not to undertake in this paper.}) a theory of distinctions: between what can and cannot
be efficiently computed, between different kinds of computational resources, and between apparently
similar models whose computational capabilities may turn out to differ in rather significant ways\footnote{Yes, I know, it sounds like a "stock LLM" introduction. I still like it though, especially after my changes and "humanization": I was writing philosophical, \emph{compare-and-contrast} introductions before the LLM era, e.g. \cite{65}.}.   
Yet not all distinctions between complexity classes are equally
informative.  To prove that a class $\mathcal D$ is not contained in a class $\mathcal C$ it is enough
to exhibit a single language in $\mathcal D\setminus\mathcal C$.  A stronger kind of separation
asks for a witness that has no infinite ``easy core'' in $\mathcal{C}$ at all.  This is the point of view embodied by
\emph{immunity}.

The notion itself migrated to Complexity Theory from (the older field of)
Recursion Theory\footnote{I prefer the old, historical  name of this area.}:  Post's construction of simple sets \cite{41} indicated 
that incomputability could be studied not only through the undecidability of a set as a whole, but also
through the absence of large recursively enumerable subsets.  The subsequent theory of immune and biimmune sets, including
the study of their degrees \cite{42}, made immunity a standard structural tool 
of recursion theory.  When imported into computational complexity (see e.g. Homer and Maass \cite{45}), the idea of immunity acquired a
resource-bounded interpretation: a language $L$ is $\mathcal C$-immune if it is infinite but contains
no infinite sublanguage in $\mathcal C$.  Thus a language being $\mathcal C$-immune witnesses  considerably more
than nonmembership in $\mathcal C$; it says that every infinite fragment of the witness is "$\mathcal{C}$-hard".

Should immunity be viewed simply as a technically stronger form of class separation? There is (some) evidence that the 
answer to the previous question is negative, and that the 
existence of 
immune languages within complexity classes is, indeed,  significant.  For instance, $P$-bi-immunity is prevalent (at least in $E$) \cite{49}, is equivalent to
a form of almost-everywhere polynomial-time hardness
\cite{43}, and implies strong distributional hardness
\cite{44,46}.  More strikingly, Pavan and Selman \cite{47}
showed that if, for some $\epsilon>0$, $\NP$ contains a
$\DTIME(2^{n^\epsilon})$-biimmune language, then natural notions of
$\NP$-completeness separate: there exists a language that is
$2$-Turing complete for $\NP$ but not $1$-truth-table complete\footnote{Hitchcock, Pavan and Vinodchandran later
obtained a related separation under the weaker hypothesis of partial
biimmunity \cite{48}.}.  In another paper \cite{54} they showed that the existence of $P$-biimmune sets has implications for average-case complexity. These results suggest
that immunity is not just a convenient strengthening of
non-membership: it captures a form of intractability with potentially important\footnote{But, sadly, somewhat forgotten/ignored.} structural consequences.

Strong separations by immune sets were investigated for relativized polynomial
hierarchies by Ko \cite{38} and Bruschi \cite{12}.  Bennett and Gill's random-oracle
work  \cite{2}, dating to what one might call \emph{the classical period of Complexity Theory}, already contained an immunity theorem at the first level; Vereshchagin
later obtained the more symmetric conclusion that, relative to a random oracle, $\NP$ and $\coNP$
contain sets immune to the opposite class \cite{3,63}.  These results suggest a natural
question: \emph{how far can one push the classical setting of strong separation into the theory of
random oracles?}

A customary inspection of the literature reveals  that probability-one separations with immunity are
substantially less common in the literature than ordinary random-oracle
separations, which have achieved a respectable degree of success (e.g. \cite{4,2,6,29,13,24,39,32,1,3}). Beyond the classical results discussed above, not that many people seemed to care\footnote{At least by writing subsequent papers on the topic.}.  A notable exception is the
work of Hemaspaandra and Zimand, who introduced stronger notions of
immunity and proved, in particular, that relative to a random oracle
$\NP^A$ and $\oplus\mathrm{P}^A$ contain $\mathrm{P}^A$-balanced-immune
sets \cite{32}.  By contrast, much of the subsequent
literature on immunity and relativization concerns specially constructed
oracles, or generic oracles, rather than probability-one statements; see, for instance,
\cite{12, 40}.  Thus the extensive theory of random-oracle
separations has, to a significant extent, developed without a parallel
theory of random-oracle separations witnessed by immune languages.

The most striking case is, perhaps, that of separating the relativized polynomial hierarchy: here separations with respect to a random oracle were achieved by the seminal result of Rossman, Servedio and Tan \cite{1} (henceforth abbreviated as RST; see also \cite{39}), while for separations with immunity there has been essentially no progress since Bennett-Gill and Vereshchagin. 

\textbf{The goal of this paper is to propose an approach that can often lift separations with respect to a random oracle to separations with immunity.} A major success of this approach is that \emph{we strengthen the results of
Bennett-Gill and Vereshchagin (and the random oracle consequence of Rossman, Servedio and Tan) by separating arbitrary levels of the polynomial hierarchy with immunity}. Our approach is more general than this single result: we demonstrate this through a number of other applications (see the next section for specific details). 

The rest of the paper is organized as follows:   Section~\ref{sec:results} formally states the results we prove.  Section~\ref{sec:ideas} discusses the common 
ideas behind our approach.  Section~\ref{sec:prelim} collects the complexity-theoretic, probabilistic,
cryptographic, and lower-bound preliminaries.  Section~\ref{sec:proofs} proves all claimed results except the main one (the separation of the polynomial hierarchy with immunity with respect to a random oracle), beginning with the
applications in which the method is most transparent. 
The proof of the main result is given in Section~\ref{sec:proofmain}, and has to clear up several hurdles: 
\begin{itemize} 
\item We note that the bottom/intermediate fan-ins of the Sipser formulas from \cite{1} are polynomial-time computable; on the other hand the top fan-in is not obviously so. Deriving the random-oracle separation from the circuit lower bound follows the framework of \cite{61} (Chapter 7), which requires the hard formulas to be \emph{uniformly constructible}; for the balanced Sipser functions of \cite{1,39,17} this is not explicitly addressed in the literature, but is quite important to ensure that the test languages we construct in the proof of Theorem~\ref{thm:main} are members of the required complexity classes. 
\item We offer a safe way to circumvent this issue by proposing a reparameterized version of Sipser formulas, and deriving a hardness result for these formulas from the corresponding hardness result in \cite{1}. 
\item Since the reparameterized formulas will have in general a different number of variables from the original one, to derive such a result we need to resort to a \emph{coupling} of the variables in the two formulas.
\item Once such a hardness result is available for the reparameterized class, the proof proceeds along avenues similar to those for the other results.    
\end{itemize} 
Finally, Section~\ref{sec:conc} discusses some limitations of our approach and subsequent open problems.

\section{Formal Statement of Our Results} 
\label{sec:results}

Our main theorem improves upon the seminal results of Bennett and Gill \cite{2} and Vereshchagin \cite{3,63}, as well as the (single oracle) corresponding strong separations of Bruschi \cite{12} and Ko \cite{38}, separating the levels of the polynomial hierarchy (with respect to a random oracle) \emph{with immunity:} 

\begin{theorem}
\label{thm:main}
With probability one over the random oracle
$A$, simultaneously for every integer $k\geq 1$:
\begin{enumerate}[label=(\alph*)]
  \item there exists an infinite language
  $L_k^A\in\SigmaP{k}{A}$ that is $\PiP{k}{A}$-immune;
  \item there exists an infinite language
  $R_k^A\in\PiP{k}{A}$ that is $\SigmaP{k}{A}$-immune.
\end{enumerate}
\end{theorem}

Relatedly, Cai \cite{6} (see also Babai \cite{5} for a simpler proof) had separated $\PH$ from $\mathsf{PSPACE}$ with respect to a random oracle. This separation also holds with immunity: 

\begin{theorem}
\label{thm:ph-pspace-immune}
With probability one over a
random oracle $A$, there exists a language
$L^A\in \mathord{\oplus}\mathsf P^A
  \subseteq \mathsf{PSPACE}^A$
that is $\mathsf{\PH}^A$-immune.  Consequently, with probability 1
\begin{equation} 
  \mathsf{PSPACE}^A\mbox{ contains a }\mathsf{PH}^A\mbox{-immune language.} 
\end{equation} 
\end{theorem}

\begin{remark} 
Regan and Royer \cite[Corollary 26]{60} observed that Cai's earlier result, together with their theorem, implies that with probability one $\PH^{A}\subsetneq \ParityP^{A}$. The previous theorem strengthens this separation to separation with immunity. 
\end{remark} 

On the other hand, Bennett and Gill \cite{2} had  claimed that with probability 1 $\mathsf{PSPACE}^A$ contains a language not in $\mathsf{PP}^{A}$. However, Beigel \cite{7} observed that their proof was not correct. The claim was finally established by Aspnes, Beigel, Furst, and Rudich \cite{4}. We improve over this last result as follows: 

\begin{theorem}\label{thm:pp:pspace}
With probability 1 over a random oracle $A$, $
 \ParityP^A\text{ contains a }\PP^A\text{-immune language}$, hence with probability 1
 \begin{equation} 
\PSPACE^A\text{ contains a }\PP^A\text{-immune language}. 
\end{equation} 
\end{theorem}

\begin{remark} 
Rothe has presented (Corollary 4.3 in \cite{40}) an oracle relative to which $\ParityP$ contains a language immune to both relativized $\PH$ and $\PP$, and has asked (in the closing lines of \cite{40}) whether any of the oracle constructions in his  paper extend to immunity with respect to a random oracle. 

The (common) language we use to prove  Theorems~\ref{thm:ph-pspace-immune} and~\ref{thm:pp:pspace} yields a positive answer to Rothe's question.  
\end{remark} 

Our next application concerns the failure of the random oracle hypothesis \cite{13}. To establish this result, Chang et al. showed that with probability 1 over a random oracle $A$ there is a language in $\coNP^A\setminus \IP^A$. The corresponding improvement is: 

\begin{theorem}\label{thm:ip}
With probability one over a random oracle $A$, there is a language
$L^A\in\coNP^A$ that is $\IP^A$-immune.  Hence with probability 1
 \begin{equation} 
\PSPACE^A\text{ contains an }\IP^A\text{-immune language}. 
\end{equation} 
\end{theorem}

Finally, separating $P^A$ and $\NP^A\cap \coNP^A$ with respect to a random oracle had long been an open problem: Bennett and Gill \cite{2} proved such a result relative to a \emph{random permutation} (see also \cite{28}). Recently Ghosal et al. \cite{24} obtained a conditional separation under the additional hypothesis 
$\UP\nsubseteq\RP$. We cannot extend this result to separations with immunity; instead, we employ a stronger hypothesis. Specifically, we prove: 

\begin{theorem}
\label{thm:npconp}
Assume that a uniform one-way permutation
exists in the sense of Definition~\ref{def:owp}.  Then, with
probability one over the random oracle $A$, there exist two disjoint
languages $
L_0^A,L_1^A\in\UP^A\cap\coUP^A$
such that
\begin{equation}
L_0^A\mathbin{\cup}L_1^A=T_{\geq 2}:= \{1^{2^{j!}}:j\geq2\}\in \Pclass
\end{equation}

and both $L_0^A$ and $L_1^A$ are $\Pclass^A$-immune. Consequently, with probability 1 there is a language in 
$\UP^A\cap\coUP^A$ (hence in $\NPclass^A\cap\coNP^A$) that is $\Pclass^{A}$-immune.
\end{theorem}

The previous series of results may have given the impression that every separation result under a random oracle extends (perhaps at the cost of stronger conditional hypotheses) to separation with immunity. This is \emph{not} true: 
Cai \cite{29} had separated the Boolean Hierarchy \cite{30} (see below) with respect to a random oracle. 

\begin{definition}

Let $A$ be an oracle.  The levels of the Boolean hierarchy relative

to $A$ are defined recursively as follows: $\BH_0^A=\Pclass^A$, 
$\BH_1^A = \NP^A$, 
and, for every $k\geq 1$,
\begin{equation} 
   \BH_{k+1}^A =
   \begin{cases}
      \{L_1\setminus L_2 :
          L_1\in \BH_k^A,\;
          L_2\in \NP^A\},
          & \text{if $k$ is odd},\\[1ex]
      \{L_1\cup L_2 :
          L_1\in \BH_k^A,\;
          L_2\in \NP^A\},
          & \text{if $k$ is even}.
   \end{cases}
\end{equation} 

Equivalently,
\begin{equation} 
   \BH_{2i}^A
      =\{L_1\cap L_2 :
           L_1\in\BH_{2i-1}^A,\;
           L_2\in\coNP^A\},
\end{equation} 

and

\begin{equation} 
   \BH_{2i+1}^A
      =\{L_1\cup L_2 :
           L_1\in\BH_{2i}^A,\;
           L_2\in\NP^A\}.
\end{equation} 

The \emph{Boolean hierarchy relative to $A$} is
$\BH^A=\bigcup_{k\geq1}\BH_k^A$. For the empty oracle we simply write $\BH_k$ and $\BH$.
In particular, $\BH_1^A=\NP^A$ and
\begin{equation} 
   \BH_2^A
      =\{L_1\cap L_2:
           L_1\in\NP^A,\;
           L_2\in\coNP^A\}.
\end{equation} 
The latter class is conventionally denoted by $\DP^A$.
\end{definition}

Cai's level-by-level random-oracle
  separation \cite{29} cannot be strengthened to level-by-level immunity beyond the base level:

\begin{observation}
\label{thm:bh}
The following statements hold: 
\begin{enumerate}[label=(\alph*),leftmargin=*]
  \item With probability one over a random oracle $A$, there is a language in
  $\DP^A$ that is $\NP^A$-immune.

  \item For every oracle $A$ and every $k\geq2$, no infinite language 
  $L\in \BH^A$ is $\BH_k^A$-immune.  
  \end{enumerate}
\end{observation}
\begin{proof}

\begin{itemize} 
\item[(a). ] By Vereshchagin's result \cite{3}, with probability 1 there is a language in $\coNP^{A}\subseteq DP^{A}$ that is $\NP^{A}$-immune. 
\item[(b). ]  Part (b) is the relativized form of
\cite[Theorem~4.3]{37}.  Alternatively, it follows
immediately from their normal-form Theorem 2.1.1: every $L\in\mathrm{BH}^A$ is a finite union of $\mathrm{DP}^A$ languages, so if $L$ is infinite, at least one of these languages is an infinite $\mathrm{DP}^A$ subset of $L$.

\end{itemize} 
\end{proof}


\section{Highlights of Our Approach}
\label{sec:ideas}

Our proofs exploit existing lower bounds on boolean functions, together with three ingredients: 
\begin{itemize} 
\item \textbf{The main ingredient is a} \emph{rare-or-wrong} \textbf{principle.}  Fix a machine from the class against which
immunity is to be proved, and consider its attempts to recognize an infinite subset of the language being constructed.  Intuitively, at each relevant input length, either the machine accepts only with sufficiently small conditional probability (the ``rare'' case), or, if it accepts often enough, the remaining randomness gives it a substantial probability of accepting an input outside the target language (the ``wrong'' case).  Thus an infinite correct subset would require the machine to avoid both alternatives indefinitely.  
\item \emph{The probabilistic mechanism excluding this possibility} is \textbf{a conditional, or L\'evy-type, version of the Borel-Cantelli lemma.}  We expose the random oracle through a filtration $(\mathcal F_s)_{s\geq 0}$ (see Section~\ref{sec:prelim} below) so that at stage $s$ the past is contained in $\mathcal F_{s-1}$ while a sufficiently large block of fresh oracle bits remains unexposed.  Conditional on the past, we estimate the probabilities of acceptance and error: if the sum of the former is finite, the machine accepts at only finitely many relevant stages, whereas if it diverges, the rare-or-wrong estimate forces the conditional probabilities of errors to have divergent sum, and Borel-Cantelli then makes errors occur
infinitely often. The mechanism that couples (globally) the good and bad events at different lengths is the summability of the difference terms between good and bad events in the "rare-or-wrong" lemma(s). 
\item Finally, \textbf{an appropriately \emph{spaced set} of relevant input lengths} separates the two roles played by the oracle.  The computation at an earlier stage is completed before the fresh randomness needed at a later stage is exposed.  This spacing prevents dependencies created by oracle queries from contaminating the next probabilistic experiment and thereby allows us to separate the complexity-theoretic part of the argument (which provides a hardness or unpredictability estimate at a single stage)  from the probabilistic part, which converts these local estimates into an almost-sure
global immunity statement.
\end{itemize} 

There are, of course, some problem-specific technical complications. They will be presented when giving each specific proof. For now we want to give some details on some of our main technical points:

\subsection{The rare-or-wrong inequalities.}

Immunity is one-sided: to preclude immunity, an algorithm need not
approximate the target on most inputs; it may accept a very sparse collection
of positive inputs.  The following elementary inequality extracts exactly the
one-sided information needed to guarantee immunity under a random oracle: 

\begin{lemma}
\label{lem:rare-or-wrong}
Let $H,C:\{0,1\}^N\to\{0,1\}$, and let $P$ be a probability on $\{0,1\}^N$.  Suppose that, for some $\delta,\varepsilon\geq 0$,
\begin{align}
  \left|\Prb_{P}(H=1)-\frac12\right|&\leq\delta,
\label{eq:rare-wrong-balance}\\
  \Prb_{P}(C=H)&\leq\frac12+\varepsilon.
\label{eq:rare-wrong-agreement}
\end{align}
Set
\[
  p=\Prb_{P}(C=1),
  \qquad
  q=\Prb_{P}(C=1\text{ and }H=0).
\]
Then
\begin{equation}
  q\geq\frac p2-\frac{\delta+\varepsilon}{2}.
\label{eq:rare-wrong-conclusion}
\end{equation}
\end{lemma}

\begin{proof}
Let $\mu=\Prb_{P}(H=1)$.  Since
\[
  \Prb_{P}(C=1,H=1)=p-q
  \quad\text{and}\quad
  \Prb_{P}(C=0,H=0)=1-\mu-q,
\]
we obtain
\[
  \Prb_{P}(C=H)=1-\mu+p-2q.
\]
Therefore
\[
  q=\frac{1-\mu+p-\Prb_{P}(C=H)}{2}.
\]
Now $1-\mu\geq\frac12-\delta$ by
\eqref{eq:rare-wrong-balance}, while
$\Prb_{P}(C=H)\leq\frac12+\varepsilon$ by
\eqref{eq:rare-wrong-agreement}.  Substitution proves
\eqref{eq:rare-wrong-conclusion}.
\end{proof}

\begin{remark} The name "rare-or-wrong" is intuitively motivated (in this case) by the following observation: for an arbitrary $\rho>0$ either $p\leq \rho$ (the "good event" is rare) or $q\geq \frac{\rho-\delta-\epsilon}{2}$ (the "bad event" is not rare). 
\end{remark} 

Lemma~\ref{lem:rare-or-wrong} will be useful in several of our results. However, when diagonalizing against probabilistic algorithms we have not one but \emph{two} thresholds to take into consideration. In such cases the following "two-threshold rare-or-wrong" lemma may prove more appropriate:  

\begin{lemma}
\label{lem:bilateral-rare-or-wrong}
Let $\Omega$ be a nonempty finite set, equipped with the uniform probability measure,
let
\[
H:\Omega\to\{0,1\},
\qquad
v:\Omega\to[0,1],
\]
and fix two thresholds $0\leq b<a\leq 1$.

Define
\[
E=\{x\in\Omega:H(x)=1,\ v(x)>a\},
\]
and
\[
B=\{y\in\Omega:H(y)=0,\ v(y)>b\}.
\]

Suppose that there is a relation $R\subseteq E\times B$ and integer constants $d,D>0$ such that:

\begin{enumerate}
\item For every $x\in E$ there are at least $d$ indices $y\in B$ such that $(x,y)\in R$; 
\item For every $y\in B$ there are at most $D$ indices $x\in E$ with $(x,y)\in R$.
\end{enumerate}

Then
\begin{equation} 
\Prb(B)\ge \frac{d}{D}\Pr(E).
\label{two-sided-conclusion}
\end{equation} 
\end{lemma} 

\begin{proof}
Double count the pairs in $R$.

By the first hypothesis, every $x\in E$ is incident to at least $d$
pairs in $R$. Hence $ 
|R|\ge d\,|E|$.  

By the second hypothesis, every $y\in B$ is incident to at most $D$
pairs in $R$. Hence $
|R|\le D\,|B|.$
Therefore $
d\,|E|\le D\,|B|$,  and thus 
\begin{equation} 
|B|\ge \frac{d}{D}|E|.
\label{ineq1}
\end{equation}

Since $\Omega$ carries the uniform probability measure,
\begin{equation} 
\Prb(E)=\frac{|E|}{|\Omega|},
\qquad
\Prb(B)=\frac{|B|}{|\Omega|}.
\end{equation} 
Hence the desired inequality~(\ref{two-sided-conclusion}) follows.  
\end{proof}

\begin{remark} In this case the "rare-or-wrong" name is motivated by the following inequality: For every $\rho>0$
\begin{equation} 
\Prb(E)\leq \rho\mbox{ ("rare event"), or }\Prb(B)> \frac{d}{D}\rho\mbox{ ("wrong event.")}
\label{row} 
\end{equation} 
\end{remark}

\begin{remark} The purpose of this lemma is to handle a randomized acceptance probability $v:\Omega\to [0,1]$, rather than a single bit. Intuitively: 
\begin{itemize}
\item[-] $H(x)=1$ will mean that $x$ belongs to the target language. 
\item[-] $v(x)$ will be the probability that the probabilistic algorithm $M$ we diagonalize against (in our specific application it will be an interactive-proof verifier) accepts $x$. 
\item[-] $a,b$ are the thresholds for acceptance/rejection of this machine. 
\item[-] The "good event" $E$ is the fact that $x$ is in the language and $M$ accepts $x$ with probability $>a$. 
\item[-] The "bad event" $B$ is the fact that $x$ is not in the language but $M$ accepts $x$ with probability $>b$. 
\item[-] The "rare-or-wrong" inequality is~(\ref{row}). 
\end{itemize} 
\end{remark}

A third ``rare or wrong'' principle is appropriate in cases when the
sample space carries an involution that flips the value of a bit-signature
$H$:  

\begin{lemma}
\label{lem:involution-rare-or-wrong}
Let $\Omega$ be a nonempty finite set equipped with the uniform
probability measure, and let $\Phi\colon\Omega\to\Omega$ be an
\emph{involution}, i.e.\ $\Phi\circ\Phi=\mathrm{id}$.  Assume that
$H\colon\Omega\to\{0,1\}$ satisfies
\begin{equation}
  H(\Phi(\omega))=1-H(\omega)
  \qquad\text{for every }\omega\in\Omega .
  \label{eq:H-flips}
\end{equation}
Let $E,T\subseteq\Omega$ be such that $T$ is $\Phi$-invariant (i.e. 
$\Phi(T)=T$), and such that%
\footnote{Intuitively, outside the exceptional set $T$, membership in $E$
is unchanged by the involution $\Phi$.}
\begin{equation}
  E\,\triangle\,\Phi(E)\ \subseteq\ T .
  \label{eq:E-stable-off-T}
\end{equation}
For $b\in\{0,1\}$ put $B_b=E\cap\{H\neq b\}$.  Then
\begin{equation}
  |B_b|\ \geq\ \tfrac12\bigl(|E|-|T|\bigr),
  \qquad\text{equivalently}\qquad
  \Prb(B_b)\ \geq\ \tfrac12\bigl(\Prb(E)-\Prb(T)\bigr).
  \label{eq:involution-bound}
\end{equation}
\end{lemma}

\begin{proof}
Put $E_0=E\setminus T$.
\begin{claim} 
\emph{$E_0$ is $\Phi$-invariant.}  
\end{claim} 
\begin{proof} 
Let $\omega\in E_0$.  Since
$\Phi(T)=T$ and $\omega\notin T$, also $\Phi(\omega)\notin T$.  Moreover
$\omega\in E$ and $\omega\notin T$, so by~\eqref{eq:E-stable-off-T}
$\omega\notin E\,\triangle\,\Phi(E)$, whence $\omega\in\Phi(E)$, say
$\omega=\Phi(e)$ with $e\in E$.  As $\Phi$ is an involution,
$\Phi(\omega)=e\in E$.  Therefore $\Phi(\omega)\in E_0$, and applying the
same argument to $\Phi(\omega)$ gives $\Phi(E_0)=E_0$.
\end{proof} 
By
\eqref{eq:H-flips}, $\Phi$ maps $E_0\cap\{H=b\}$ into
$E_0\cap\{H\neq b\}$ and vice versa; being an involution, it is a
bijection between these two sets.  Hence
\begin{equation}
  \bigl|E_0\cap\{H=b\}\bigr|=\bigl|E_0\cap\{H\neq b\}\bigr| .
  \label{eq:pairing}
\end{equation}
The two sets partition $E_0$, so each has exactly $\tfrac12|E_0|$
elements.

Since $E_0\cap\{H\neq b\}\subseteq B_b$,
\[
  |B_b|\ \geq\ \bigl|E_0\cap\{H\neq b\}\bigr|
         \ =\ \tfrac12|E_0|
         \ \geq\ \tfrac12\bigl(|E|-|T|\bigr).
\]
Dividing by $|\Omega|$ gives the probabilistic form
of~\eqref{eq:involution-bound}.
\end{proof}

\begin{remark}
As before, the name is explained by a disjunctive reformulation: in the
setting of Lemma~\ref{lem:involution-rare-or-wrong}, for every $\rho>0$
either
\[
  \Prb(E)\leq\Prb(T)+\rho\ \ 
  \qquad\text{or}\qquad
  \Prb(B_b)>\tfrac{\rho}{2}\ \ 
\]
That is, intuitively, either the basic event $E$ is rare, or a random sample
$\omega$ from $E$ is often ``wrong'' (i.e.\ $H(\omega)$ differs from a
target value $b\in\{0,1\}$).  
It is, however, the inequality~\eqref{eq:involution-bound} that we apply.
\label{rem:involution-rare-or-wrong}
\end{remark}

\subsection{The spacing of input lengths.}

A standard random-oracle construction assigns a hard target bit to every
input length, schematically
\begin{equation}
  L_{\mathrm{std}}^A
  =\left\{1^n:F_n\bigl(A\upharpoonright\{0,1\}^n\bigr)=1\right\}.
\label{eq:standard-oracle-language}
\end{equation}
For ordinary separations, this strategy is often enough to force a disagreement, hence separation.  Immunity
requires more: whenever a machine continues to accept designated inputs with
nonsummable conditional probability, one must force a false positive.

On input $1^n$, a polynomial-time oracle machine
may query strings whose lengths are as large as $q(n)$ for some polynomial
$q$.  Revealing enough of the oracle to determine its behavior at length $n$
can therefore reveal the target blocks for many larger input lengths.  

We avoid this problem by using only the lengths
\begin{equation}
  t_j=2^{j!},
  \qquad j\geq 1.
\label{eq:tower}
\end{equation}
They satisfy
\begin{equation}
  t_{j+1}=t_j^{j+1}.
\label{eq:tower-recurrence}
\end{equation}
The target languages constructed below will be subsets of $T=\{1^{t_j}:j\geq 1\}.$ 
This spacing has two properties: 

\begin{lemma}
\label{lem:support-decidable}
The unary language $T$ is decidable in deterministic polynomial time.
\end{lemma}

\begin{proof}
On input $1^n$, first test whether $n$ is a power of two.  If not, reject.  If
$n=2^r$, compute $1!,2!,3!,\ldots$ until the first value that is at least $r$,
and accept exactly when one of these values equals $r$.  All integers involved
have $O(\log n)$ bits, and even an implementation that scans the unary input
runs in time polynomial in $n$.
\end{proof}

\begin{lemma}
\label{lem:wide-spacing}
For every polynomial $q$, there is $j_0$ such that for all $j\geq j_0$ 
\begin{equation} 
  q(t_j)<t_{j+1}
\end{equation} 
\end{lemma}

\begin{proof}
Choose $c\geq 1$ and $n_0$ such that $q(n)\leq n^c$ for all $n\geq n_0$.
Choose $j_1$ such that $t_j\geq n_0$ for $j\geq j_1$.  If
$j\geq\max\{c,j_1\}$, then
$q(t_j)\leq t_j^c<t_j^{j+1}=t_{j+1}.$
\end{proof}

The inequality $q(t_j)<t_{j+1}$ enables the two-step exposure approach used later in the proofs:
first reveal every relevant oracle coordinate except the current block, and
then reveal the current block. 

\section{Preliminaries} 
\label{sec:prelim}

\subsection{Notions from Computational Complexity.} 

We will assume familiarity with standard notions in complexity theory as presented e.g. in \cite{23}, including finite alphabets, (oracle) Turing machines, languages, complexity classes such as $\Pclass,\NP,\Sigma_{k}^{P},\Pi_{k}^{P}, \Delta_{k}^{P},\PH,\PSPACE,\IP,\PP, \ParityP$. For an oracle $A$, the corresponding relativized classes are denoted by using $A$ as a superscript (e.g. 
$\SigmaP{k}{A}$, $\PiP{k}{A}$, and $\DeltaP{k}{A}$). 

We fix the following acceptance convention for $\mathsf{PP}$:  A nondeterministic polynomial-time oracle machine $M$ \emph{accepts $x$ relative to $A$} if \emph{at least as many computation paths of $M^A(x)$ accept as reject}; $L(M^A)$ is the set of such $x$, and $\mathsf{PP}^A$ is the class of languages of this form. The convention is harmless: if $M$ instead uses the strict rule, that more paths accept than reject, replace $M$ by the machine that first makes one nondeterministic binary choice, on one branch simulating $M$ and on the other branch rejecting outright after a single step. This adds one rejecting path, so ties under the new rule correspond to strict majorities under the old one, and the two machines accept the same language relative to every oracle, since the added branch makes no oracle query.

We will use notation 
$\mathrm{bin}_\ell(k)$ to denote the $\ell$-bit binary representation of an integer $k$, for $0 \le k < 2^\ell$. 

\begin{definition} 
Let $\Omega=\bits^{\bits^*}$
be the space of all Boolean oracles.  We equip $\Omega$ with the product
probability measure under which the coordinates $A(z)$, $z\in\bits^*$, are
independent unbiased bits.  A statement holds \emph{with respect to a random
oracle} or \emph{almost surely} if it holds for a set of oracles of measure
one.
\end{definition} 

\begin{definition}
Let $\mathcal C$ be a class of languages.  A language $L$ is
\emph{$\mathcal C$-immune} if
\begin{enumerate}[label=(\roman*)]
  \item $L$ is infinite, and
  \item every language $S\in\mathcal C$ satisfying $S\subseteq L$ is finite.
\end{enumerate}
A language $L$ is \emph{$\mathcal C$-biimmune} if both $L$ and its complement are $\mathcal C$-immune. 
\end{definition}

To prove Theorem~\ref{thm:ip} we need  \emph{interactive proof systems} \cite{15,16,21,22}. We review in the sequel only the specific concepts we will need: 

\begin{definition} 
A \emph{(clocked probabilistic oracle) verifier} is a probabilistic
oracle Turing machine $V$ together with an explicit polynomial clock.
On an input of length $n$, the clock bounds the running time, the
number of oracle queries, every oracle-query length, every message
length, and the total number of rounds by a polynomial in $n$.

A \emph{prover} is an arbitrary, computationally unbounded party interacting with $V$; a \emph{prover strategy $P$} is a function mapping each partial transcript to the next prover message, and we allow $P$'s choice to depend on the input $x$ and on the oracle $A$.
\label{def:ips}
\end{definition}

The following approach to defining acceptance of a language agrees with the approach in \cite{13} (see e.g. Section 3.1): 

\begin{definition} 
Fix a verifier $V$ such as in Definition~\ref{def:ips}. For an oracle $A$ and input $x$, define
\begin{equation}
\val_V(A,x)
=
\max_P
\Prob[
V^A\text{ accepts }x\text{ after interacting with }P
],
\label{eq:value}
\end{equation}
where the probability is over the private randomness of $V$ and the
maximum is over all prover strategies.
\end{definition} 

\begin{observation} 
The maximum from equation~(\ref{eq:value}) exists (see \cite{13}, from where this convention is imported). 
\end{observation}

\begin{definition} 
A verifier $V^A$ is \emph{an interactive proof verifier for a language $S$} if, for all words $x$, it holds that 
\begin{align}
x\in S
&\Longrightarrow
\val_V(A,x)>\frac23,
\label{eq:complete}\\
x\notin S
&\Longrightarrow
\val_V(A,x)<\frac13.
\label{eq:sound}
\end{align}
The constants $\frac{2}{3}$ and $\frac{1}{3}$ are a normalization: for any oracle $A$, a verifier with completeness $c$ and soundness $s$ satisfying $c-s \ge 1/p(n)$ for a polynomial $p$ can be amplified, by polynomially many independent repetitions against the same oracle 
followed by an acceptance threshold strictly between \(s\) and \(c\), for example \((c+s)/2\), into a clocked verifier meeting~(\refeq{eq:complete}) and~(\refeq{eq:sound}), and recognizing the same language. We may therefore restrict attention throughout to verifiers of this form.
\end{definition} 

\begin{definition} 
A \emph{polynomial threshold function} is a real function $f:\{0,1\}^N \to \{0,1\}$ having the form $f(x)=sign(q(x))$
where $q(\cdot)$ is a real polynomial and $sign$ is the function 
\[
sign(x)=\left\{\begin{array}{ll}
0, & \mbox{ if }x<0\\
1, & \mbox{ otherwise.}\\
\end{array} 
\right. 
\]
\end{definition} 

Finally, we assume knowledge of the approach of Rossman, Servedio and Tan (for a readable introduction see \cite{17}). 

For integers $d,m\geq 2$, a depth-$d$ bottom fan-in $m$ \emph{Sipser formula} is a
read-once monotone formula 
\begin{equation}
   \Sipser_{d,m} :
   \{0,1\}^{N_{RST}(d,m)} \longrightarrow \{0,1\}
   \label{eq:sipser-family}
\end{equation}

consisting of $d$ alternating layers of
unbounded-fan-in AND and OR gates.  The bottom gates are AND gates of fan-in $m$.  All gates at a given layer
have the same type and the same fan-in, consecutive layers have opposite types, and every
input variable labels exactly one leaf.  The fan-ins are chosen as follows (indexed from the root): 
\begin{equation}
  w_0,\underbrace{w,\ldots,w}_{d-2\text{ intermediate levels}}, m.
  \label{eq:rst-fanins}
\end{equation}

We next review the parameters chosen in  \cite{1} for the function $Sipser_{d,m}$\footnote{We refer to the full version of \cite{1}. More precisely, the definitions of
$\rho$ and $w$ below are respectively equations~(2) and~(3) of \cite[Section~6, p.~15]{1}, while the definitions
of $q$ and $\lambda$ are given in equation~(7) of
\cite[Section~7.1, p.~16]{1}. The journal version \cite{39} uses slightly different values of the parameters.}. Indeed, put
\begin{equation}
  \rho=2^{-m},
  \qquad
  q=\sqrt\rho=2^{-m/2}.
  \label{eq:rho-q}
\end{equation}
The intermediate fan-in, as chosen in \cite{1}, is
\begin{equation}
  w=w(m)=\left\lfloor\frac{m2^m}{\log_{2}(e)}\right\rfloor
   =\left\lfloor m2^m\ln2\right\rfloor.
  \label{eq:w-definition}
\end{equation}
RST define
\begin{equation}
  \lambda=\frac{(\log_{2} w)^{3/2}}{w^{5/4}},
  \label{eq:lambda}
\end{equation}
and recursively define
\begin{align}
  t_{d-1}&=\frac{\rho-\lambda}{q},
  \label{eq:last-t}\\
  t_{r-1}&=\frac{(1-t_r)^{qw}-\lambda}{q},
  \qquad r=d-1,d-2,\ldots,2.
  \label{eq:t-recursion}
\end{align}
Let 
\begin{equation}
   D_m=-q\ln(1-t_1),
   \qquad
   R_m=\frac{\ln2}{D_m}.
\end{equation}
The top fan-in is chosen in \cite{1} as 
\begin{equation}
  w_0=\min\left\{s\in\mathbb N:(1-t_1)^{qs}\leq\frac12\right\}
     =\left\lceil R_m\right\rceil
  \label{eq:w0-definition}
\end{equation}
The number of variables of the boolean function $Sipser_{d,m}$ is thus 
\begin{equation}
  N=N_{RST}(d,m)=w_0\cdot w^{d-2}\cdot m. 
  \label{eq:N-definition}
\end{equation}

For
fixed $d\geq 2$, the RST choice of fan-ins satisfies
\begin{equation}
   \log_2 N_{RST}(d,m)=\Theta_d(m).
   \label{eq:sipser-size-growth}
\end{equation}
In particular, there are constants $c_d,C_d>0$ such that, for all
sufficiently large $m$,
\begin{equation}
   2^{c_d\cdot m}\leq N_{RST}(d,m)\leq 2^{C_d\cdot m}.
   \label{eq:sipser-N-bounds}
\end{equation}


\subsection{Notions from Probability Theory.} 

We assume familiarity with introductory probability theory, including probability spaces, $\sigma$-fields, the classical Borel-Cantelli lemmas, at the level of, say, \cite{50}. To make the paper self-contained, we recall the following: 

\begin{lemma}[Second Borel-Cantelli lemma]
\label{lem:second-borel-cantelli}
Let $(D_n)_{n\geq 1}$ be a sequence of mutually independent events
in a probability space. If
\begin{equation} 
    \sum_{n=1}^{\infty} \Prb(D_n)=\infty,
\end{equation} 
then
\begin{equation} 
    \Prb(D_n \text{ infinitely often})=1.
\end{equation} 
\end{lemma}

We will also need a standard (e.g. \cite{31}, Chapter 5, pp. 95–102) conditional variant of Borel-Cantelli:

\begin{definition}
Let $(\Omega,\mathcal F,\mathbb P)$ be a probability space. A \emph{filtration} is an increasing sequence of sub-$\sigma$-fields
\begin{equation} 
   \mathcal F_0\subseteq \mathcal F_1\subseteq \mathcal F_2
   \subseteq\cdots\subseteq\mathcal F.
\end{equation} 


A sequence of events $(E_n)_{n\geq 1}$ is said to be \emph{adapted} to the filtration $(\mathcal F_n)_{n\geq 0}$ if $E_n\in\mathcal F_n$ for every $n\geq 1$. In this case, the conditional probability
$\Prb(E_n\mid\mathcal F_{n-1})$ 
can be interpreted as the probability that $E_n$ occurs given all the information revealed before stage $n$.
\end{definition}

\begin{lemma}[L\'evy's extension of the Borel-Cantelli lemmas]
\label{lem:CBC}
Let $(\mathcal H_n)_{n\geq 0}$ be a filtration, and let
$D_n\in\mathcal H_n$ for every $n\geq 1$. Then
\[
  \Prb\left(
    \{D_n\ \mathrm{i.o.}\}
    \mathbin{\triangle}
    \left\{
      \sum_{n=1}^{\infty}
      \Prb(D_n\mid\mathcal H_{n-1})=\infty
    \right\}
  \right)=0,
\]
where $\triangle$ denotes symmetric difference.
In particular, almost surely, if the conditional-probability sum is finite, only finitely
many $D_n$ occur; if it is infinite, infinitely many occur.
\end{lemma}

\subsection{Notions from Cryptography.}

We use some notions from cryptography as well \cite{25,36}, specifically one-way
permutations. There are several inequivalent such notions in the literature; see e.g. Goldreich's discussion
\cite{26}.  We therefore state explicitly the
specific type we need. 

We first recall the following 
\begin{definition} 
A function $\nu:\mathbb N\to\mathbb R_{\geq 0}$ is
\emph{negligible} (e.g. \cite{36}) if, for every polynomial $p$, there exists $m_0$ such
that
\[
\nu(m)<\frac{1}{p(m)}
\]
for every $m\geq m_0$.
\end{definition} 

\begin{definition}
\label{def:owp}
A \emph{uniform one-way permutation} is a deterministic polynomial-time
algorithm $F$ such that, for every $m\geq1$,
$f_m(x):=F(1^m,x)$
is a permutation $
f_m:\bits^m\longrightarrow\bits^m$ 
and for every probabilistic polynomial-time algorithm $I(\cdot,\cdot;\cdot)$
there is a negligible function $\nu_I$ such that 
\begin{equation}
\Prb_{x\leftarrow U_m, r}
\left[
f_m\bigl(I(1^m,f_m(x);r)\bigr)=f_m(x)
\right]
\leq \nu_I(m).
\label{2.1} 
\end{equation}
Here $U_m$ denotes the uniform distribution on $\bits^m$ and $r$ denotes the internal randomness of $I$. To allow for the possibility that $I$ fails (denoted $\perp$), define $f_m(\perp)=\perp$. 
\end{definition}

Since $f_m$ is a permutation, $f_m(x)$ is uniform when $x$ is
uniform.  Thus (\ref{2.1}) is equivalently the requirement
\begin{equation}
\Prb_{y\leftarrow U_m, r}
\left[
f_m\bigl(I(1^m,y;r)\bigr)=y
\right]
\leq \nu_I(m).
\end{equation}

\begin{remark} 

Merely assuming that no
deterministic polynomial-time algorithm computes the inverse
everywhere seems insufficient for the proof.  We need the inversion
probability on a uniformly random image to be negligible. 
\end{remark} 

Length-preserving one-to-one one-way functions have been constructed
from standard number-theoretic assumptions such as RSA- and
discrete-logarithm-type assumptions; see Goldreich, Levin, and Nisan
\cite{27}.  Because a length-preserving injection
from $\bits^m$ to itself is automatically surjective, such a function
is a permutation at each length.


\begin{remark}
The hypothesis in Definition~\ref{def:owp} is at least as strong as the
worst-case complexity-theoretic existence of a one-to-one one-way
function.  It is also at least as strong (and probably stronger) than the hypothesis
$\UP\nsubseteq\RP$ used by Ghosal et al (whose goal was to avoid one-way permutations).  The latter follows already
from suitable injective one-way functions
\cite{24}.
\end{remark}


\subsection{Hardness Results.} 

Our strong separation results exploit various hardness results from the literature (some of them tailored to our purpose). We review in the sequel the specific form of (some of) these hardness results.  

First, the following correlation bound on parity is due to H{\aa}stad (e.g. \cite{34}): 

\begin{lemma}[Average-case parity lower bound]
\label{lem:parity-correlation-input}
For every fixed depth $d\geq 2$, there is a constant $k_d>0$ such that
the following holds.  Let $C:\{0,1\}^N\to\{0,1\}$ be an unbounded-fan-in
AND/OR/NOT circuit of depth at most $d$ and size at most $S$.  Then
\begin{equation}
  \left|
    \Prb_X[C(X)=\operatorname{PARITY}_N(X)]-\frac12
  \right|
  \leq
  2^{-k_d N/(\log_{2}(2S))^{d-1}}.
\label{eq:parity-correlation-bound}
\end{equation}
Here $X$ is uniformly distributed on $\{0,1\}^N$.
\end{lemma}

We use the following consequence of the main voting-polynomial lower bound
of~\cite{4}: 

\begin{lemma}
\label{lem:abfr}
Let $p$ be a real polynomial of degree at most $d$ in $N$ Boolean
variables, and let $g(x)=sign[p(x)]$.
If $\operatorname{PARITY}_N$ denotes parity on $\{0,1\}^N$, then
\begin{equation} 
 \Prb_{x\in\{0,1\}^N}
 \bigl[g(x)\neq \operatorname{PARITY}_N(x)\bigr]
 \geq
 2^{-N}\sum_{i=0}^{\lfloor (N-d-1)/2\rfloor}\binom Ni .
 \label{foo4} 
\end{equation} 
Consequently there is an absolute constant $C>0$ such that, whenever
$d<N$,
\begin{equation} 
 \Prb[g(x)\neq\operatorname{PARITY}_N(x)]
 \geq \frac12-C\frac{d+1}{\sqrt N}.
 \label{foo2} 
\end{equation} 
\end{lemma}

\begin{proof}
The first assertion is the parity specialization of the ABFR
voting-polynomial lower bound; parity has weak degree $N$.
For the second assertion, compare the sum in \eqref{foo4} with the lower half of the
binomial distribution.  Moving the cutoff from the center by $O(d+1)$
removes at most $O(d+1)$ central binomial coefficients.  Since
\[
 2^{-N}\binom{N}{\lfloor N/2\rfloor}=O(N^{-1/2}),
\]
(\ref{foo2}) follows.
\end{proof}

Finally, we use the following hardness result from \cite{1} (Theorem~7 from the full version of \cite{1}, reindexed by the bottom-width parameter $m$):

\begin{definition}\label{def:normal-form}
A Boolean circuit is in \emph{normal form} if it is a layered alternating AND/OR circuit
(in the sense of \cite{1}) whose leaves are literals or the constants $0,1$; in
particular, NOT gates occur only on input literals. The root gate determines the
alternation pattern. The size of a circuit is its number of gates.
\end{definition}

Every circuit can be converted into normal form without increasing its depth, blowing up
its size from $S$ to $S^{O_d(1)}+O_d(N)$, by pushing negations to the inputs, merging
consecutive gates of the same type and inserting unary dummy gates.
\begin{remark}\label{rem:constants}
Constant leaves can be eliminated: Repeatedly replace a gate having an absorbing constant
input ($0$ for AND, $1$ for OR) by that constant, delete neutral constant inputs
($1$ for AND, $0$ for OR), and treat a gate left without inputs as the corresponding
constant. This only deletes gates, preserves layering, does not increase depth, and
preserves the root type unless the root itself becomes constant.  In the latter
case, replace the circuit by a depth-$2$ constant gadget with the required root type:
$
\mathrm{OR}(\mathrm{AND}(x),\mathrm{AND}(\neg x))\equiv 1$, $
\mathrm{AND}(\mathrm{OR}(x,\neg x))\equiv 1$, $\mathrm{AND}(\mathrm{OR}(x),\mathrm{OR}(\neg x))\equiv 0$, $
\mathrm{OR}(\mathrm{AND}(x,\neg x))\equiv 0$. 
Hence every normal-form circuit $C$ of depth at most $d\ge 2$ is equivalent to a
constant-free normal-form circuit with the same root type, depth at most $d$, and size
at most $\max(\mathrm{size}(C),3)$. This is used in Corollary~\ref{cor:rst} below to extend Proposition~\ref{thm:RST} to normal-form circuits with constant leaves.
\end{remark}

\begin{proposition}[\cite{1}, Theorem 7]\label{thm:RST}
For every fixed $d\ge 2$ there is $b_d>0$ such that, for all sufficiently large $m$, every
constant-free normal-form circuit $C$ of depth exactly $d$, with root gate opposite to that
of $\mathrm{Sipser}_{d,m}$ and $\mathrm{size}(C)\le 2^{N^{1/(6(d-1))}}$, satisfies
\[
\Prb_{X\in\{0,1\}^N}\bigl[C(X)=\mathrm{Sipser}_{d,m}(X)\bigr]\le \tfrac12+N^{-b_d}.
\]
\end{proposition}

\begin{corollary}\label{cor:rst}
Let $a_d:=1/(7(d-1))$. For all sufficiently large $m$, the conclusion of
Proposition~\ref{thm:RST} holds for every normal-form circuit $C$ (constant leaves
allowed) of depth at most $d$, with root gate opposite to that of $\mathrm{Sipser}_{d,m}$
and $\mathrm{size}(C)\le 2^{N^{a_d}}$.
\end{corollary}

\begin{proof}
By Remark~\ref{rem:constants}, $C$ is equivalent to a constant-free normal-form circuit $C^{\prime}$ with the same
root type, depth at most $d$, and $\mathrm{size}(C^{\prime})\le\max(\mathrm{size}(C),3)$. Pad $C^{\prime}$
to depth exactly $d$ without changing its root type: add two unary gates above the root
(a new root of the same type over a unary gate of the opposite type) as often as needed,
and, if the parity of the remaining depth difference is odd, insert one unary gate of the
opposite type between a deepest gate and one of its literal inputs. This adds at most $d$
gates. Since $2^{N^{a_d}}+d+3\le 2^{N^{1/(6(d-1))}}$ for all sufficiently large $m$,
Proposition~\ref{thm:RST} applies to the padded circuit.
\end{proof}




The proof of Theorem~\ref{thm:main} below will employ this bound to derive a specially-tailored hardness lemma. We point to Lemma~\ref{lem:effective-sipser} below for details.

\section{Proofs of Our Other Results}
\label{sec:proofs} 

In this section we prove the results stated in  Section~\ref{sec:results}, except for the main theorem, which gets a separate section.  We do not
follow the order of their presentation.  Instead, we begin with the case in which the method from Section~\ref{sec:ideas} appears in its simplest
form, and then treat the cases requiring additional ingredients.  

\subsection{Proof of Theorem~\ref{thm:ph-pspace-immune}.}

We begin with Theorem~\ref{thm:ph-pspace-immune}.  This proof gives the cleanest
illustration of the rare-or-wrong argument and of the delayed exposure of the
random oracle.

For each $j\geq 1$, let $V_j=\{0,1\}^{t_j}.$ 
The sets $V_j$ are pairwise disjoint blocks of oracle coordinates.  Define
\begin{equation}
  Y_j^A=\bigoplus_{y\in V_j}A(y),
\label{eq:block-parity}
\end{equation}
and let
\begin{equation}
  L^A=\left\{1^{t_j}:Y_j^A=1\right\}.
\label{eq:parity-language}
\end{equation}
Thus $L^A$ is supported on the sparse unary set $T$ introduced in
Section~\ref{sec:ideas}.

\begin{lemma}
\label{lem:parity-language-membership}
For every oracle $A$,
\[
  L^A\in\ParityP^A\subseteq\mathsf{PSPACE}^A.
\]
\end{lemma}

\begin{proof}
On input $x$, first verify deterministically that $x=1^n$ and that $1^n\in T$.
If either test fails, reject.  Otherwise, use exactly $n$ binary
nondeterministic choices to guess $y\in\{0,1\}^n$, query $A(y)$, and accept
that computation path exactly when $A(y)=1$.  There is one computation path
for each $y\in\{0,1\}^n$.  Hence the parity of the number of accepting paths is
\[
  \bigoplus_{y\in\{0,1\}^n}A(y),
\]
which is the membership bit in~\eqref{eq:parity-language}.  Therefore
$L^A\in\ParityP^A$.

\end{proof}

\begin{lemma}
\label{lem:parity-language-infinite}
With probability one over a random oracle $A$, the language $L^A$ is infinite.
\end{lemma}

\begin{proof}
For every $j$, the random variable $Y_j^A$ is unbiased, since it is the parity
of a nonempty family of independent unbiased oracle bits.  Moreover,
$Y_1^A,Y_2^A,\ldots$ are mutually independent because the blocks
$V_1,V_2,\ldots$ are disjoint.  Consequently,
\[
  \Prb_A(Y_j^A=1)=\frac12
\]
for every $j$.  The second Borel--Cantelli lemma therefore gives
\[
  \Prb_A(Y_j^A=1\text{ for infinitely many }j)=1.
\]
By~\eqref{eq:parity-language}, this is precisely the assertion that $L^A$ is
infinite almost surely.
\end{proof}

\begin{lemma}[Summability of the parity-correlation errors]
\label{lem:eta-summable}
Fix constants $c\geq 1$, $d\geq 2$, and $\kappa>0$.  Then
\[
  \sum_{j\geq 1}
  2^{-\kappa\cdot 2^{t_j}/(1+t_j^c)^{d-1}}
  <\infty.
\]
\end{lemma}

\begin{proof}
Set $K=c(d-1)$.  Since $t_j\geq 1$,
$
  (1+t_j^c)^{d-1}\leq 2^{d-1}t_j^K.$
Thus, with $\kappa'=\kappa/2^{d-1}>0$,
$2^{-\kappa \cdot 2^{t_j}/(1+t_j^c)^{d-1}}
  \leq
  2^{-\kappa'\cdot 2^{t_j}/t_j^K}.$
The quantity $2^m/m^K$ eventually dominates $m/\kappa'$.  Hence, for all
sufficiently large $j$,
\begin{equation} 
  \kappa'\cdot \frac{2^{t_j}}{t_j^K}\geq t_j,
\end{equation} 
and the $j$-th summand is at most $2^{-t_j}\leq 2^{-j}$.  The tail is therefore
dominated by a geometric series.
\end{proof}

Let $(M_e)_{e\geq 1}$ enumerate all syntactically clocked alternating oracle
machines having a constant number of alternations.  For every oracle $A$, the
languages $L(M_e^A)$ range over $\mathsf{PH}^A$.  It remains to prove the
following machine-wise assertion.

\begin{lemma}
\label{lem:fixed-ph-machine}
For every $e\geq 1$,
\begin{equation}
  \Prb_A\left[
    L(M_e^A)\nsubseteq L^A
    \text{ or }
    L(M_e^A)\text{ is finite}
  \right]=1.
\label{eq:fixed-ph-machine}
\end{equation}
\end{lemma}

\begin{proof}
Fix $e$ and write $M=M_e$.  Since $M$ belongs to a fixed level of the
polynomial hierarchy, there are constants $d\geq 2$ and $c\geq 1$ such that,
on every input $x$ of length $n$, the acceptance predicate of $M^A(x)$ can be
represented by an unbounded-fan-in Boolean circuit $C_x$ in the relevant
oracle variables satisfying
\begin{equation}
  \operatorname{depth}(C_x)\leq d,
  \qquad
  \operatorname{size}(C_x)\leq 2^{n^c}.
\label{eq:ph-circuit-parameters}
\end{equation}
There is also a polynomial $q$ such that every oracle variable occurring in
$C_x$ corresponds to a query string of length at most $q(n)$.

One way to obtain this representation is to start from a quantified
characterization of the language decided by $M$.  Expand each polynomially
bounded existential quantifier into an unbounded-fan-in OR gate and each
universal quantifier into an unbounded-fan-in AND gate.  The remaining
deterministic polynomial-time oracle computation can be expanded into its
oracle decision tree and expressed as a DNF or CNF.  Because $M$ is fixed,
this adds only a constant number of layers and produces a circuit of size
$2^{n^{O(1)}}$; increasing $c$ if necessary gives~\eqref{eq:ph-circuit-parameters}.

By the spacing property of $(t_j)$ established in Section~\ref{sec:ideas},
there is $j_0$ such that
\begin{equation}
  q(t_j)<t_{j+1}
\label{eq:tower-gap-pspace}
\end{equation}
for every $j\geq j_0$.  The finitely many indices below $j_0$ do not affect
summability or any infinitely-often event, so all subsequent definitions and
arguments are restricted to $j\geq j_0$.

Define
\begin{align}
  \mathcal F_j
  &=\sigma\bigl(A(z):|z|<t_{j+1}\bigr),
\label{eq:def-Fj}\\
  \mathcal G_j
  &=\sigma\bigl(A(z):|z|<t_{j+1},\ z\notin V_j\bigr).
\label{eq:def-Gj}
\end{align}
Then
\begin{equation}
  \mathcal F_{j-1}\subseteq\mathcal G_j\subseteq\mathcal F_j.
\label{eq:interleaving-sigma-fields}
\end{equation}
Let
\begin{equation}
  E_j=\{M^A(1^{t_j})\text{ accepts}\}.
\label{eq:def-Ej}
\end{equation}
By~\eqref{eq:tower-gap-pspace}, every query made by $M$ on $1^{t_j}$ has
length less than $t_{j+1}$, and hence $E_j\in\mathcal F_j$.

Condition on an assignment $\omega$ to the oracle coordinates exposed by
$\mathcal G_j$.  The only still-unexposed coordinates of length less than
$t_{j+1}$ are those in $V_j$.  Therefore $A\upharpoonright V_j$ remains
uniform on $\{0,1\}^{N_j}$, where
\begin{equation}
  N_j=|V_j|=2^{t_j}.
\end{equation} 
Under this conditioning, the acceptance predicate of $M^A(1^{t_j})$ becomes a
fixed circuit $
  C_{j,\omega}:\{0,1\}^{V_j}\longrightarrow\{0,1\}$
of depth at most $d$ and size at most $2^{t_j^c}$.  Restricting the other
oracle variables cannot increase either parameter.  The target bit $Y_j^A$
remains the parity of all $N_j$ fresh variables.

Lemma~\ref{lem:parity-correlation-input}, applied on every atom of
$\mathcal G_j$, gives
\begin{equation}
  a_j:=\Pr\bigl(\mathbf 1_{E_j}=Y_j^A\mid\mathcal G_j\bigr)
  \leq \frac12+\eta_j,
\label{eq:conditional-parity-agreement}
\end{equation}
where, after decreasing the positive constant $k_d$ in the exponent, if necessary, to a lower value $\kappa_d$ we may take
\begin{equation}
  \eta_j
  =2^{-\kappa_d 2^{t_j}/(1+t_j^c)^{d-1}}
\label{eq:eta-parity}
\end{equation}
for a constant $\kappa_d>0$ depending only on $d$.

Lemma~\ref{lem:eta-summable} shows that
\begin{equation}
  \sum_{j\geq j_0}\eta_j<\infty.
\label{eq:eta-summable}
\end{equation}

Define the $\mathcal G_j$-measurable conditional probabilities
\begin{align}
  p_j&=\Pr(E_j\mid\mathcal G_j),
\label{eq:def-pj}\\
  b_j&=\Pr\bigl(E_j\cap\{Y_j^A=0\}\mid\mathcal G_j\bigr).
\label{eq:def-bj}
\end{align}
Since $\mathcal G_j$ is generated by finitely many oracle coordinates, we may
work atom by atom.  For each assignment $\omega$ exposed by $\mathcal G_j$,
consider the conditional probability space obtained by fixing that assignment.  In
Lemma~\ref{lem:rare-or-wrong}, take
$H=Y_j^A$, $C=\mathbf 1_{E_j}$, $\delta=0$, $\varepsilon=\eta_j$. 
The parity bit is exactly balanced in this conditional space, so
\begin{equation} 
  \left|\Prb(H=1)-\frac12\right|=0\leq\eta_j,
\end{equation} 
and~\eqref{eq:conditional-parity-agreement} gives
$\Pr(C=H)\leq\frac12+\eta_j$.  Lemma~\ref{lem:rare-or-wrong} therefore yields,
on every atom of $\mathcal G_j$,
\begin{equation}
  b_j\geq\frac12p_j-\frac12\eta_j.
\label{eq:direct-rare-or-wrong}
\end{equation}

Let
\begin{equation}
  B_j=E_j\cap\{Y_j^A=0\}.
\label{eq:def-Bj}
\end{equation}
The event $B_j$ is a false positive: $M$ accepts $1^{t_j}$ although
$1^{t_j}\notin L^A$.  Since both $E_j$ and $Y_j^A$ are
$\mathcal F_j$-measurable, $B_j\in\mathcal F_j$, and
\begin{equation} 
  \Prb(B_j\mid\mathcal G_j)=b_j.
\end{equation} 

We now apply Lemma~\ref{lem:CBC} to an explicit interleaving of the two
sigma-fields.  Define a filtration $(\mathcal H_n)$ by
\begin{equation}
  \mathcal H_{2j-1}=\mathcal G_j,
  \qquad
  \mathcal H_{2j}=\mathcal F_j
  \qquad (j\geq j_0),
\label{eq:interleaved-filtration}
\end{equation}
and fill the finitely many earlier terms with $\mathcal F_{j_0-1}$.  Set all
corresponding earlier events to $\varnothing$.  The
inclusions in~\eqref{eq:interleaving-sigma-fields} show that this is indeed a
filtration.  For the acceptance events, set
\begin{equation} 
  D^{E}_{2j-1}=\varnothing,
  \qquad
  D^{E}_{2j}=E_j.
\end{equation} 
Then
\begin{equation} 
  \Pr(D^{E}_{2j}\mid\mathcal H_{2j-1})=p_j,
\end{equation} 
so Lemma~\ref{lem:CBC} gives
\begin{equation}
  E_j\text{ occurs infinitely often}
  \quad\Longleftrightarrow\quad
  \sum_{j\geq j_0}p_j=\infty
\label{eq:cbc-E}
\end{equation}
almost surely.  Applying the same construction with
$D^{B}_{2j}=B_j$ and $D^{B}_{2j-1}=\varnothing$ gives
\begin{equation}
  B_j\text{ occurs infinitely often}
  \quad\Longleftrightarrow\quad
  \sum_{j\geq j_0}b_j=\infty
\label{eq:cbc-B}
\end{equation}
almost surely.

Work now on the probability-one event on which both equivalences hold.  
\begin{itemize} 

\item If
$\sum_jp_j<\infty$, then~\eqref{eq:cbc-E} implies that $M$ accepts only
finitely many strings of the form $1^{t_j}$.  Since $L^A$ contains no strings
outside this support, the implication
\[
  L(M^A)\subseteq L^A
  \quad\Longrightarrow\quad
  L(M^A)\text{ is finite}
\]
follows.

\item If instead $\sum_jp_j=\infty$, then~\eqref{eq:direct-rare-or-wrong} and
Lemma~\ref{lem:eta-summable} imply
\begin{equation} 
  \sum_jb_j
  \geq
  \frac12\sum_jp_j-\frac12\sum_j\eta_j
  =\infty.
\end{equation} 
By~\eqref{eq:cbc-B}, the false-positive events $B_j$ occur infinitely often.
In particular, $M$ accepts a string outside $L^A$, and hence
\[
  L(M^A)\nsubseteq L^A.
\]
\end{itemize} 
The two cases prove~\eqref{eq:fixed-ph-machine}.
\end{proof}

There are countably many machines in the enumeration.  Intersecting the
probability-one events from Lemma~\ref{lem:fixed-ph-machine} over all $e$, and
then intersecting with the probability-one event from
Lemma~\ref{lem:parity-language-infinite}, shows that, almost surely, $L^A$ is
infinite and has no infinite $\mathsf{PH}^A$ subset.  Thus $L^A$ is
$\mathsf{PH}^A$-immune almost surely.  Lemma~\ref{lem:parity-language-membership}
now completes the proof of Theorem~\ref{thm:ph-pspace-immune}.

\subsection{Proof of Theorem~\ref{thm:pp:pspace}.}


The intuition behind the proof of Theorem~\ref{thm:pp:pspace} is that a PP machine’s decision on the fresh oracle block is a low-degree polynomial threshold function; by Lemma~\ref{lem:abfr}, such a function cannot predict parity, so it must either accept too rarely to define an infinite sublanguage or make false-positive errors infinitely often. The L\'evy extension of the Borel-Cantelli lemma converts this intuition (via an application of Lemma~\ref{lem:rare-or-wrong}) into a proof of immunity.

Indeed, fix an enumeration $(M_e)_{e\geq 1}$ of  polynomial-time oracle nondeterministic machines.  We use the same language $L^A$ as in the previous proof. 

\begin{lemma} 
For every $e\geq 1$
\begin{equation} 
 \Prb_A\!\left[
 L(M_{e}^A)\subseteq L^A
 \ \text{and}\ 
 L(M_{e}^A)\text{ is infinite}
 \right]=0.
\label{eq:50}
\end{equation} 
\end{lemma} 
\begin{proof} 

There is a polynomial $q$ such that, on inputs of length
$n$, every computation path runs for at most $q(n)$ steps and hence asks at
most $q(n)$ oracle queries, all of length at most $q(n)$. Without loss of generality we can choose $q$ so that it is monotonically increasing and $q(n)\geq n$. 

By the spacing property there exists $j_0\geq 1$ such that 
for $j\geq j_0$ we have $q(t_j)<t_{j+1}$.  Put
\begin{equation} 
 X_j=\{M_{e}^A(1^{t_j})\text{ accepts in the }\PP\text{ sense}\}
\end{equation} 
and let
\begin{equation} 
 Y_j=\{1^{t_j}\in L^A\},\qquad
 B_j=X_j\cap Y_j^{c}.
\end{equation} 
Thus $B_j$ is the event that $M_e$ makes a false positive with respect to
$L^A$ at the $j$-th special length.

Let
$V_j=\{0,1\}^{t_j}$ be the $j$-th block of oracle coordinates and write
\begin{equation}
  Z_j=\bigl(A(z):z\in V_j\bigr)\in\{0,1\}^{N_j},
  \qquad N_j=2^{t_j},
  \label{eq:fresh-block-PP}
\end{equation}
so that $Y_j$ is the event that the parity of $Z_j$ equals $1$.
Define the two $\sigma$-fields
\begin{align}
  \mathcal F_j&=\sigma\bigl(A(z):|z|<t_{j+1}\bigr),
  \label{eq:F-PP}\\
  \mathcal G_j&=\sigma\bigl(A(z):|z|<t_{j+1},\ z\notin V_j\bigr).
  \label{eq:G-PP}
\end{align}
Since $V_j$ consists of strings of length exactly $t_j$, while
$\mathcal F_{j-1}$ exposes only strings of length $<t_j$, we have
\begin{equation}
  \mathcal F_{j-1}\subseteq\mathcal G_j\subseteq\mathcal F_j
  \subseteq\mathcal G_{j+1}\subseteq\cdots,
  \label{eq:interleaved-PP}
\end{equation}
an increasing interleaved filtration.  By the spacing property every
oracle query made by $M_e$ on $1^{t_j}$ has length at most
$q(t_j)<t_{j+1}$; hence $X_j\in\mathcal F_j$, and $Y_j\in\mathcal F_j$
because $Y_j$ is determined by $Z_j$.  Consequently
$B_j=X_j\cap Y_j^{c}\in\mathcal F_j$ as well.

The point of using $\mathcal G_j$ rather than
$\sigma\bigl(A(z):|z|\le q(t_j)\bigr)$ is the following.  Fix an atom
$\omega$ of $\mathcal G_j$, that is, an assignment to all oracle
coordinates of length $<t_{j+1}$ outside $V_j$.  Every coordinate that
$M_e^A(1^{t_j})$ can query is then either frozen by $\omega$ or belongs
to $V_j$: there are no further unexposed coordinates that could affect
the computation.  Moreover, conditionally on $\omega$ the vector $Z_j$ is
uniformly distributed on $\{0,1\}^{N_j}$.  We may therefore work atom by
atom, and on each atom the acceptance predicate is a function of the
fresh block alone.  We record the standard representation of that
function.

We need a standard representation lemma: 

\begin{lemma}
\label{lem:gap}
Fix $j\geq j_0$ and an atom $\omega$ of $\mathcal G_j$.  Viewed as a
function of the fresh block $Z_j$, the predicate $X_j$ is the sign of a
real multilinear polynomial $p_{j,\omega}(Z_j)$ of degree at most
$d_j=q(t_j)$.
\end{lemma}

\begin{proof}
Consider a computation path of $M_e$ on $1^{t_j}$.  Once its internal
nondeterministic choices are fixed, the condition that this path is
followed can be expressed as a product of literals corresponding to the
answers to its adaptive oracle queries.  

Since at most $q(t_j)$ oracle queries occur on a path, the resulting path-indicator polynomial has degree at most $q(t_j)$. (A path may repeat a query; the corresponding literal then occurs with a power, which may be reduced using $z^2 = z$ on $\{0,1\}$ for each of the $N_j$ potential variables without changing the function or raising the degree.)

Sum these indicators over accepting paths, and subtract the corresponding
sum over rejecting paths.  The resulting gap polynomial has degree at most
$q(t_j)$: summing/subtracting over the exponentially many paths does not increase the degree: the gap polynomial has degree at most $q(t_j)$ regardless of how many paths there are.

On the other hand the \emph{sign} of the gap polynomial determines whether the accepting-path count is at least 
the rejecting-path count.  
\end{proof}

Fix an atom $\omega$ of $\mathcal G_j$.  By
Lemma~\ref{lem:gap}, $X_j$ is represented there by the sign of
a polynomial of degree at most $d_j=q(t_j)$ in the $N_j$ fresh bits,
while $Y_j$ is the parity of those same bits.  Lemma~\ref{lem:abfr}
therefore gives
\begin{equation}
  \Pr\bigl(X_j\,\triangle\,Y_j\mid\mathcal G_j\bigr)\ \geq\ \frac12-\varepsilon_j,
  \qquad
  \varepsilon_j=C\,\frac{q(t_j)+1}{2^{t_j/2}} ,
  \label{eq:ABFR-applied}
\end{equation}
on every atom of $\mathcal G_j$.  Applying Lemma~\ref{lem:rare-or-wrong}
on that atom with $H=\mathbf 1_{Y_j}$, $C=\mathbf 1_{X_j}$, $\delta=0$
(the parity of a nonempty block of fresh unbiased bits is exactly
balanced) and $\varepsilon=\varepsilon_j$, we obtain
\begin{equation}
  \Pr\bigl(B_j\mid\mathcal G_j\bigr)
  \ \geq\ \frac{\Pr\bigl(X_j\mid\mathcal G_j\bigr)-\varepsilon_j}{2}.
  \label{eq:PP-rare-or-wrong}
\end{equation}
Since the atom was arbitrary and $\mathcal G_j$ is generated by finitely
many oracle coordinates, \eqref{eq:PP-rare-or-wrong} holds everywhere.
Ignoring the finitely many indices $j\le j_0$, put
\begin{equation}
  Q_j=\Pr\bigl(X_j\mid\mathcal G_j\bigr),
  \qquad
  R_j=\Pr\bigl(B_j\mid\mathcal G_j\bigr),
  \label{eq:QR-PP}
\end{equation}
so that \eqref{eq:PP-rare-or-wrong} reads
$R_j\ \geq\ (Q_j-\varepsilon_j)/2$.

\begin{lemma}
We have: 
\begin{equation} 
 \sum_j\varepsilon_j<\infty,
\label{eq:62}
\end{equation} 
\label{foo-d}
\end{lemma}
\begin{proof} 

Since $q$ is a polynomial, for all sufficiently large $j$, 
$C(q(t_{j})+1)\leq 2^{t_{j}/4}$, 
so $\epsilon_j\leq 2^{-t_{j}/4}\leq 2^{-j/4}$. Hence 
$\sum_j\varepsilon_j<\infty.$
\end{proof}

As in the proof of Theorem~\ref{thm:ph-pspace-immune}, we feed
\eqref{eq:interleaved-PP} to Lévy's extension of the Borel-Cantelli
lemmas.  Define a filtration $(\mathcal H_n)$ by
\begin{equation}
  \mathcal H_{2j-1}=\mathcal G_j,
  \qquad
  \mathcal H_{2j}=\mathcal F_j
  \qquad(j\geq j_0),
  \label{eq:H-PP}
\end{equation}
filling the finitely many earlier terms with $\mathcal F_{j_0-1}$ and the
corresponding events with $\varnothing$; by \eqref{eq:interleaved-PP}
this is indeed a filtration.  Setting
$D^X_{2j-1}=\varnothing$, $D^X_{2j}=X_j$ gives
$\Pr(D^X_{2j}\mid\mathcal H_{2j-1})=Q_j$, and likewise
$D^B_{2j-1}=\varnothing$, $D^B_{2j}=B_j$ gives
$\Pr(D^B_{2j}\mid\mathcal H_{2j-1})=R_j$.  Lemma~\ref{lem:CBC} therefore
yields, almost surely,
\begin{equation}
  X_j\ \text{i.o.}\iff\sum_j Q_j=\infty,
  \qquad
  B_j\ \text{i.o.}\iff\sum_j R_j=\infty,
  \label{eq:levy-PP}
\end{equation}
and we may complete the proof by the same two-case analysis as in the
proof of Theorem~\ref{thm:ph-pspace-immune}:

\begin{itemize} 
\item{Case 1: $\sum_j Q_j<\infty$}. 

Then Lemma~\ref{lem:CBC}, applied to $X_j$, implies that almost surely $X_j$ occurs only
finitely often.  Thus $M_e^A$ accepts only finitely many of the
candidate strings $1^{t_j}$.

\item{Case 2: $\sum_j Q_j=\infty$.} 

Since \(R_j\geq (Q_j-\varepsilon_j)/2\), while
\(\sum_j Q_j=\infty\) and \(\sum_j\varepsilon_j<\infty\) by Lemma~\ref{foo-d} it follows  that
$\sum_j R_j=\infty$. 
Applying Lemma~\ref{lem:CBC} to $B_j$ shows that $B_j$ occurs infinitely
often almost surely.  In particular, $L(M_e^A)$ cannot be a subset of $L^A$.
\end{itemize} 
\end{proof} 

There are only countably many polynomial-time oracle machines (including
all polynomial clocks), so the union of the
exceptional null sets in (\ref{eq:50}) is null.  Therefore, with probability one,
no infinite $\PP^A$ language is a subset of $L^A$, in other words (we already proved that $L^A$ is almost surely infinite) $L^A$ is $\PP^A$-immune.  

Since (we showed already that)  $L^A\in\ParityP^A\subseteq\PSPACE^A$, the theorem follows.

\subsection{Proof of Theorem~\ref{thm:ip}.}

The construction of the proof partitions a portion of the oracle at each designated length into \emph{exponentially many disjoint segments}, and declares the corresponding unary string to belong to $L^A$ exactly when none of these segments is completely filled with ones; this makes $L^A$ a $\mathrm{coNP}^A$ language and leaves it infinite almost surely. Now fix an interactive-proof verifier $V$. If $V$ accepts a positive instance with probability greater than $2/3$, then, since every computation examines only polynomially many of the exponentially many segments, most segments are queried with probability less than $1/3$. Completing any such segment produces a negative instance that $V$ still accepts with probability greater than $1/3$.

Double counting these modifications as in  Lemma~\ref{lem:bilateral-rare-or-wrong} allows converting this dichotomy (via the conditional Borel-Cantelli lemma) into the conclusion of  $\mathrm{IP}^A$-immunity.  

Indeed, fix an even integer $n\geq4$ and define
\begin{equation}
s(n)=\frac n2,
\qquad
M(n)=2^{s(n)}.
\label{eq:sM}
\end{equation}

Order $\bits^n$ lexicographically, equivalently by interpreting
$n$-bit strings as integers in $\{0,\ldots,2^n-1\}$. For
$i\in\{0,\ldots,M(n)-1\}$, define the segment
\begin{equation}
S_i(n)
=
\left\{
\operatorname{bin}_n(i\,s(n)+r):
0\leq r<s(n)
\right\}.
\label{eq:segments}
\end{equation}

Since
\begin{equation}
M(n)s(n)
=
2^{n/2}\frac n2
\leq 2^n
\end{equation}
for all sufficiently large $n$, the segments are well defined,
pairwise disjoint, and contained in $\bits^n$.

Let
\begin{equation}
W(n)
=
\bigcup_{i=0}^{M(n)-1}S_i(n).
\label{eq:W}
\end{equation}

The segment addressing is uniform: given $n$, $i$, and $r$, the
string $\operatorname{bin}_n(i\,s(n)+r)$ is computable in time
polynomial in $n$.

For an oracle $A$, define
\begin{equation}
L^A
=
\left\{
1^{t_j}: j\geq 2, 
\forall i<M(t_j),\;
S_i(t_j)\nsubseteq A
\right\}.
\label{eq:L}
\end{equation}

\begin{lemma}
\label{prop:conp}
For every oracle $A$,
\begin{equation}
L^A\in\coNP^A.
\label{eq:conp}
\end{equation}
\end{lemma}

\begin{proof} We prove instead that $\overline{L^{A}}$, the complement of the language, is in $\NPclass^A$.

Indeed, on an input $x$, first test whether $x=1^n$ and whether $1^n\in T_{\geq 2}$.
If not, accept, since all off-support strings lie in
$\overline{L^A}$.

Suppose $n=t_j$. A nondeterministic machine guesses an index
$i<M(n)=2^{n/2}$. The index has $n/2$ bits. It then queries $A$ on all strings in $S_i(n)$ and accepts if every answer is
$1$. Thus the machine accepts exactly when some segment is fully contained
in $A$, i.e. exactly when $1^n\notin L^A$.
\end{proof}

\begin{lemma}
\label{lem:positiveprob}
We have
\begin{equation}
\lim_{1^n\in T,n\to \infty} \Prob_A[1^n\in L^A]= e^{-1}.
\end{equation}
\end{lemma}

\begin{proof}
For a fixed support length $1^n\in T$, a segment $S_i(n)$ is completely
contained in a random oracle with probability
\begin{equation}
2^{-s(n)}=\frac1{M(n)}.
\end{equation}
Since the segments are disjoint, these events are mutually
independent over $i$. Hence
\begin{equation}
\Prb_A[1^n\in L^A]
=
\left(1-\frac1{M(n)}\right)^{M(n)}\longrightarrow e^{-1}
\label{eq:positiveprob}
\end{equation}

as $n\to\infty$.
\end{proof}

\begin{lemma}
\label{prop:infinite}
With probability one, $L^A$ is infinite.
\end{lemma}

\begin{proof}
Same proof approach, based on the second Borel-Cantelli lemma, as for Lemma~\ref{lem:parity-language-infinite}. 
\end{proof}

Fix a clocked verifier $V$. Let $q(n)$ be a polynomial that
simultaneously bounds:
\begin{enumerate}[label=(\roman*)]
\item the number of oracle queries made by $V$ on inputs of length
$n$, and
\item the length of every oracle query made by $V$ on inputs of
length $n$.
\end{enumerate}

Fix a sufficiently large support length $n$ and abbreviate
\begin{equation}
s=s(n)=\frac n2,
\qquad
M=M(n)=2^s,
\qquad
W=W(n).
\end{equation}

Fix an arbitrary assignment $\omega$ to all oracle coordinates
relevant to the computation except the coordinates in $W$. For an
assignment $\beta\in\bits^W$, let $A_{\omega,\beta}$ denote the
resulting oracle on the finite query cone, and define
\begin{equation}
v_\omega(\beta)
=
\val_V(A_{\omega,\beta},1^n).
\label{eq:vb}
\end{equation}

Call $\beta$ \emph{still unfilling} if no segment $S_i(n)$ is entirely
filled with $1$s. Call $\beta$ \emph{uniquely filling} if exactly
one segment is entirely filled with $1$s.
Define
\begin{align}
\mathcal E_\omega
&=
\left\{
\beta:
\beta\text{ is still unfilling and }
v_\omega(\beta)>\frac23
\right\},
\label{eq:Eblock}\\
\mathcal B_\omega
&=
\left\{
\gamma:
\gamma\text{ is uniquely filling and }
v_\omega(\gamma)>\frac13
\right\}.
\label{eq:Bblock}
\end{align}

Intuitively, we will exploit the fact that  
the verifier can afford to watch only polynomially many of the $2^{n/2}$ segments; filling an unwatched one changes membership without changing the verifier's behavior, and the uniqueness of the filled segment is what lets each resulting counterexample be traced back to few sources.

To be able to implement the double-counting principle underlying Lemma~\ref{lem:bilateral-rare-or-wrong}, we need a relation, together with two degree guarantees. The following result provides them: 

\begin{definition} Fix $\beta\in\mathcal E_\omega$. Let $P_\beta$ be an optimal prover
strategy for $V^{A_{\omega,\beta}}$ on input $1^n$.
For each segment $S_i(n)$, define
\begin{equation}
\theta_i(\beta)
=
\Prob[
\text{during the interaction with }P_\beta,
V^{A_{\omega,\beta}}(1^n)
\text{ queries at least one string in }S_i(n)
].
\label{eq:theta}
\end{equation}

Also, given an index $i$, let $\beta^{(i)}$ be obtained from $\beta$ by
setting every bit of $S_i(n)$ equal to $1$ and leaving all other
coordinates unchanged.
\end{definition} 

\begin{lemma}[Segment-completion relation] 
\label{lem:completion} 
For every sufficiently large $n$ and every outside assignment $\omega$, define

\begin{equation} 
R_\omega = \bigl\{(\beta,\gamma) \in E_\omega \times \{0,1\}^{W} \;:\; \gamma = \beta^{(i)} \text{ for some } i < M \text{ with } \theta_i(\beta) < \tfrac13\bigr\}.
\end{equation} 

Then $R_\omega \subseteq E_\omega \times B_\omega$, and
\begin{itemize} 
\item[-] every $\beta \in E_\omega$ has at least $M - 3q(n)$ distinct $R_\omega$-neighbours in $B_\omega$;
\item[-] every $\gamma \in B_\omega$ has at most $M-1$ distinct $R_\omega$-preimages in $E_\omega$.
\end{itemize} 
\end{lemma} 
\begin{proof} 

Because the segments are disjoint and every computation path makes at
most $q(n)$ oracle queries, the number of distinct segments touched
on any computation path is at most $q(n)$. Taking expectations,
\begin{equation}
\sum_{i=0}^{M-1}\theta_i(\beta)
\leq q(n).
\label{eq:sumtheta}
\end{equation}

Consequently, at most $3q(n)$ indices can satisfy
$\theta_i(\beta)\geq1/3$. Thus at least $
M-3q(n)$ 
indices satisfy $
\theta_i(\beta)<\frac13.$

Since $\beta$ is still unfilling, no segment is full under $\beta$. Because
the segments are disjoint and only $S_i(n)$ is modified,
$\beta^{(i)}$ is uniquely filled, with $S_i(n)$ as its unique full
segment.

Now run against $A_{\omega,\beta^{(i)}}$ the same prover strategy
$P_\beta$. Couple the two verifier computations by using the same
random coin string. Until the verifier queries a string in $S_i(n)$,
the oracle answers are identical and hence the two transcripts
coincide. Therefore the acceptance probability can decrease by at
most the probability that a query enters $S_i(n)$:
\begin{equation}
\val_V(A_{\omega,\beta^{(i)}},1^n)
\geq
v_\omega(\beta)-\theta_i(\beta).
\label{eq:stability}
\end{equation}
Here the left side is the optimal value for the modified oracle, so it
is at least the acceptance probability obtained by reusing
$P_\beta$.

Since $\beta\in\mathcal E_\omega$ and $i$ satisfies
$\theta_i(\beta)<\frac13$, we have
\begin{equation}
v_\omega(\beta)>\frac23
\end{equation}
Equation~\eqref{eq:stability} therefore gives
\begin{equation}
\val_V(A_{\omega,\beta^{(i)}},1^n)>\frac13.
\end{equation}
Hence $\beta^{(i)}\in\mathcal B_\omega$.

The argument above shows that every $\beta\in E_\omega$ has at least
$M-3q(n)$ distinct $R_\omega$-neighbors in $B_\omega$. On the other hand, fix $\gamma\in B_\omega$, and let $S_i(n)$ be its
unique full segment.  If $\gamma=\beta^{(i)}$ for some 
$\beta\in E_\omega$, then $\beta$ agrees with $\gamma$ outside $S_i(n)$, whereas
inside $S_i(n)$ it can be any proper subset.  Hence $\gamma$ has at
most $2^s-1=M-1$
preimages under $R_\omega$.

\end{proof} 

\begin{lemma}[One-sided segment-completion lemma] 
\label{lem:onesided}
For every sufficiently large $n$ and every outside assignment $\omega$,

\begin{equation} 
\frac{|B_\omega|}{2^{|W|}} \;\ge\; c_V(n)\,\frac{|E_\omega|}{2^{|W|}}, \qquad c_V(n) = \frac{M - 3q(n)}{M-1},
\end{equation} 

and $c_V(n) \to 1$, so $c_V(n) \ge \frac12$ for all sufficiently large $n$.
\end{lemma} 
\begin{proof}
For $\beta\in\mathcal E_\omega$, let
\[
  \deg^{+}_{R_\omega}(\beta)
  =
  \bigl|\{\gamma\in\mathcal B_\omega:
          (\beta,\gamma)\in R_\omega\}\bigr|,
\]
and, for $\gamma\in\mathcal B_\omega$, let
\[
  \deg^{-}_{R_\omega}(\gamma)
  =
  \bigl|\{\beta\in\mathcal E_\omega:
          (\beta,\gamma)\in R_\omega\}\bigr|.
\]
Double-counting the pairs in $R_\omega$ gives
\[
  |R_\omega|
  =
  \sum_{\beta\in\mathcal E_\omega}
       \deg^{+}_{R_\omega}(\beta)
  =
  \sum_{\gamma\in\mathcal B_\omega}
       \deg^{-}_{R_\omega}(\gamma).
\]
By the first degree bound in Lemma~\ref{lem:completion},
every $\beta\in\mathcal E_\omega$ has at least
$M-3q(n)$ distinct $R_\omega$-neighbours. Hence
\[
  |R_\omega|
  \ge
  \bigl(M-3q(n)\bigr)|\mathcal E_\omega|.
\]
By the second degree bound in that lemma, every
$\gamma\in\mathcal B_\omega$ has at most $M-1$
$R_\omega$-preimages. Therefore
\[
  |R_\omega|
  \le
  (M-1)|\mathcal B_\omega|.
\]
Combining these inequalities yields
\[
  |\mathcal B_\omega|
  \ge
  \frac{M-3q(n)}{M-1}\,|\mathcal E_\omega|
  =
  c_V(n)|\mathcal E_\omega|.
\]
Dividing by the total number $2^{|W|}$ of assignments to the
coordinates in $W$ gives
\[
  \frac{|\mathcal B_\omega|}{2^{|W|}}
  \ge
  c_V(n)\frac{|\mathcal E_\omega|}{2^{|W|}},
\]
as claimed.

Finally, since $M=2^{n/2}$ and $q$ is a polynomial,
\[
  c_V(n)
  =
  \frac{M-3q(n)}{M-1}
  =
  1-\frac{3q(n)-1}{M-1}
  \longrightarrow 1.
\]
In particular, $c_V(n)\geq 1/2$ for all sufficiently large $n$.
\end{proof}


By Lemma~\ref{lem:wide-spacing}, after discarding finitely many support indices
we may assume $q(t_{j})<t_{j+1}$
for every remaining $j$.

Define
\begin{equation}
\mathcal F_j
=
\sigma\bigl(
A(z):|z|<t_{j+1}
\bigr).
\label{eq:Fj}
\end{equation}

Let $W_j=W(t_j)$ be the union of the current segments, and define
\begin{equation}
\mathcal G_j
=
\sigma\bigl(
A(z):
|z|<t_{j+1},\ z\notin W_j
\bigr).
\label{eq:Gj}
\end{equation}

Because $W_j\subseteq\bits^{t_j}$,
\begin{equation}
\mathcal F_{j-1}
\subseteq
\mathcal G_j
\subseteq
\mathcal F_j.
\label{eq:FGF}
\end{equation}
Moreover, $
\mathcal F_j\subseteq\mathcal G_{j+1}$, 
because $\mathcal F_j$ contains only strings of length strictly less
than $t_{j+1}$, whereas $W_{j+1}\subseteq\bits^{t_{j+1}}$.

Thus
\begin{equation}
\mathcal F_{j-1}
\subseteq
\mathcal G_j
\subseteq
\mathcal F_j
\subseteq
\mathcal G_{j+1}
\subseteq\cdots
\label{eq:interleave}
\end{equation}
is an increasing interleaved filtration.

Conditioned on $\mathcal G_j$, the coordinates in $W_j$ remain
independent unbiased bits.

Define the events
\begin{align}
E_j
&=
\left\{
1^{t_j}\in L^A
\text{ and }
\val_V(A,1^{t_j})>\frac23
\right\},
\label{eq:Ej}\\
B_j
&=
\left\{
\begin{array}{l}
\text{exactly one segment }S_i(t_j)\text{ is contained in }A,\\
\val_V(A,1^{t_j})>\frac13
\end{array}
\right\}.
\label{eq:Bj}
\end{align}

By Lemma~\ref{lem:wide-spacing}, for $j\geq j_0$ the value $\val_V(A,1^{t_j})$ is determined by
oracle coordinates of length strictly less than $t_{j+1}$, so
$E_j,B_j\in\mathcal F_j$.

Set
\begin{equation}
P_j=\Prb(E_j\mid\mathcal G_j),
\qquad
Q_j=\Prb(B_j\mid\mathcal G_j).
\label{eq:PQ}
\end{equation}

\begin{lemma}[Conditional one-sided estimate]
\label{lem:conditional}
For all sufficiently large $j$,
\begin{equation}
Q_j
\geq
c_V(t_j)P_j
\label{eq:condbound}
\end{equation}
almost surely. In particular, eventually
\begin{equation}
Q_j\geq\frac12P_j.
\label{eq:condhalf}
\end{equation}
\end{lemma}

\begin{proof}
Condition on an arbitrary atom of $\mathcal G_j$. All relevant oracle
bits outside $W_j$ are then fixed, while the bits of $W_j$ remain
uniformly and independently distributed.

The conditional event $E_j$ is exactly the set $\mathcal E_\omega$
of Lemma~\ref{lem:onesided}, and the conditional event $B_j$ is
exactly $\mathcal B_\omega$. The inequality follows pointwise from
Lemma~\ref{lem:onesided}.
\end{proof}

To apply Lemma~\ref{lem:CBC}, introduce the interleaved filtration
\begin{align}
\mathcal H_{2j-2}=\mathcal F_{j-1},
\mathcal H_{2j-1}=\mathcal G_j,
\mathcal H_{2j}=\mathcal F_j,
\end{align}
after reindexing from a sufficiently large initial $j$.
For the events $E_j$, define
\begin{equation}
D^E_{2j-1}=\varnothing,
\qquad
D^E_{2j}=E_j.
\end{equation}
Then
\begin{equation}
\Prb(D^E_{2j}\mid\mathcal H_{2j-1})
=
\Prb(E_j\mid\mathcal G_j)
=
P_j.
\end{equation}
Conditional Borel-Cantelli therefore gives a probability-one event
$\Omega_V^E$ on which
\begin{equation}
E_j\ \ii
\quad\Longleftrightarrow\quad
\sum_jP_j=\infty.
\label{eq:Ecbc}
\end{equation}

Similarly, with
\begin{equation}
D^B_{2j-1}=\varnothing,
\qquad
D^B_{2j}=B_j,
\end{equation}
there is a probability-one event $\Omega_V^B$ on which
\begin{equation}
B_j\ \ii
\quad\Longleftrightarrow\quad
\sum_jQ_j=\infty.
\label{eq:Bcbc}
\end{equation}

Let $\Omega_V=\Omega_V^E\cap\Omega_V^B$.
Then $\Prob(\Omega_V)=1$.

\begin{lemma}
\label{prop:fixedV}
Fix a clocked probabilistic polynomial-time oracle verifier $V$.
There is a probability-one set of oracles $\Omega_V$ such that, for
every $A\in\Omega_V$, every language $S$ for which $V^A$ is a valid
interactive proof verifier and which satisfies
$S\subseteq L^A$ 
is finite.
\end{lemma}

\begin{proof}
Fix $A\in\Omega_V$. The quantities $P_j(A)$ and $Q_j(A)$ are now
ordinary nonnegative real numbers.

There are two exhaustive cases.

\medskip
\noindent
\textbf{Case 1.}
Suppose
\begin{equation}
\sum_jP_j(A)<\infty.
\label{eq:case1}
\end{equation}
By~\eqref{eq:Ecbc}, only finitely many events $E_j(A)$ occur.

Let $S$ be any language for which $V^A$ is a valid interactive proof
verifier and suppose $S\subseteq L^A$.

If $1^{t_j}\in S$, then, because $S\subseteq L^A$,
$1^{t_j}\in L^A$. By completeness~\eqref{eq:complete},
$\val_V(A,1^{t_j})>2/3$. Hence $E_j(A)$ occurs.

Therefore only finitely many support strings can belong to $S$.
Since $L^A$ contains no off-support strings and $S\subseteq L^A$,
the language $S$ is finite.

\medskip
\noindent
\textbf{Case 2.}
Suppose
\begin{equation}
\sum_jP_j(A)=\infty.
\label{eq:case2}
\end{equation}
By Lemma~\ref{lem:conditional}, eventually
\begin{equation}
Q_j(A)\geq\frac12P_j(A).
\end{equation}
Therefore
\begin{equation}
\sum_jQ_j(A)=\infty.
\end{equation}
Equation~\eqref{eq:Bcbc} implies that $B_j(A)$ occurs infinitely
often.

In particular, for some $j$, exactly one segment at length $t_j$ is
fully contained in $A$, and
\begin{equation}
\val_V(A,1^{t_j})>\frac13.
\label{eq:badvalue}
\end{equation}
The existence of a full segment means $1^{t_j}\notin L^A$. If
$S\subseteq L^A$, then $1^{t_j}\notin S$. But
soundness~\eqref{eq:sound} for the interactive proof $V^A$ would
require
\begin{equation}
\val_V(A,1^{t_j})<\frac13,
\end{equation}
contradicting~\eqref{eq:badvalue}.

Thus Case~2 is impossible whenever $V^A$ is a valid interactive proof
verifier for a sublanguage of $L^A$.

We conclude that every such sublanguage is finite.
\end{proof}

Enumerate the clocked polynomial-time oracle verifiers; intersect the probability-one sets supplied by Lemma~\ref{prop:fixedV} and the probability-one event on which $L^A$ is infinite; then no infinite $\IP^A$ sublanguage of $L^A$ exists, while $L^A\in\coNP^A$.

\subsection{Proof of Theorem~\ref{thm:npconp}.}

The proof of 
Theorem~\ref{thm:npconp} is best understood as a cryptographic version of the “rare-or-wrong” strategy.
It uses a cryptographically hard function to choose, for every test length, a specific oracle query. It then constructs two complementary languages based on the two possible answers to this hidden query, and notes that:  
\begin{itemize} 
\item[-] an unambiguous machine can locate the relevant oracle bit by guessing its unique address; this puts the two resulting languages in relativized $UP\cap coUP$. 
\item[-] a deterministic polynomial-time machine can only locate that address with negligible probability, because doing so would invert a one-way permutation;
\item[-] without querying that address, the machine cannot distinguish whether the target bit is \(0\) or \(1\).
\item[-] the flipping involution associated with flipping the hidden bit shows that accepting without locating it has a \(1/2\)-type chance of being wrong; conditional Borel-Cantelli turns this local uncertainty (via Lemma~\ref{lem:involution-rare-or-wrong}) into almost-sure immunity.
\end{itemize} 

Indeed, fix $j\geq2$ and define
\begin{equation} 
n_j=t_j,
\qquad
m_j=\frac{n_j}{4}.
\end{equation} 

\begin{remark} 
Clearly (by the definition of $t_j$) $m_j$ is an integer for $j\geq2$.
\end{remark} 

We define two explicitly computable sets of oracle coordinates of
length exactly $n_j$.

For $1\leq r\leq m_j$, let
\begin{equation}
u_{j,r}
=
00\,
\operatorname{bin}_{m_j}(r)\,
0^{\,n_j-m_j-2},
\label{4.2}
\end{equation}

For $x\in\bits^{m_j}$, let
\begin{equation}
v_{j,x}
=
01\,
x\,
0^{\,n_j-m_j-2}.
\end{equation}

The two families of words are disjoint, and words corresponding to
different $j$'s are different because they have different lengths.

Define the word 
\begin{equation}
Y_j(A)
=
A(u_{j,1})A(u_{j,2})\cdots A(u_{j,m_j})
\in\bits^{m_j}.
\end{equation}
Since the oracle coordinates in (\ref{4.2}) are independent fair bits,
$Y_j(A)$ is uniformly distributed over $\bits^{m_j}$.
Because $f_{m_j}$ is a permutation, there is a unique
string $X_j(A)$ that is the preimage of $Y_{j}(A)$, formally 
\begin{equation} 
X_j(A)=f_{m_j}^{-1}(Y_j(A)).
\end{equation} 

Make the following notations: 
$Z_j(A)=v_{j,X_j(A)}$, 
and $
H_j(A)=A(Z_j(A))$. 

For $b\in\bits$, define
\begin{equation}
L_b^A
=
\left\{
1^{t_j}, j\geq 2:H_j(A)=b
\right\}.
\label{4.8}
\end{equation}

Clearly, $
L_0^A\mathbin{\cup}L_1^A=T_{\geq 2}, L_0^A\cap L_1^A=\emptyset$. 

\begin{lemma}
\label{prop:upcoup}
For every oracle $A$ and every $b\in\bits$,
\begin{equation}
L_b^A\in\UP^A\cap\coUP^A.
\end{equation} 
\end{lemma}

\begin{proof}
Consider first membership in $L_b^A$.

On an input $w$, deterministically check whether
$w=1^{t_j}$
for some $j\geq2$.  If not, reject. Suppose $w=1^{t_j}$.  Compute $m_j=t_j/4$ and nondeterministically
guess a string $
x\in\bits^{m_j}.$
Query the $m_j$ oracle coordinates $u_{j,r}$ and reconstruct
$Y_j(A)$.  Compute $f_{m_j}(x)$ and check whether
$f_{m_j}(x)=Y_j(A).$ 
Since $f_{m_j}$ is a permutation, exactly one guessed string can pass
this test, namely $x=X_j(A).$ 
On this unique surviving branch, query $A(v_{j,x})$ and accept if and
only if $A(v_{j,x})=b$. 

Thus there is either exactly one accepting branch or no accepting
branch.  Hence $L_b^A\in\UP^A.$

For the complement, if $w$ is not in the unary support $T_{\geq 2}$, accept
deterministically.  If $w=1^{t_j}$, perform the same unique-witness
procedure and accept on the unique valid preimage branch if and only
if $A(v_{j,x})=1-b.$ 
So $L_b^A\in\coUP^A.$
\end{proof}


\begin{lemma}
\label{lem:fair}
For every $j\geq 2$ and $b\in\bits$, $
\Prb_A[H_j=b]=\frac12.$ 
Moreover, the random variables $H_j$, $j\geq2$, are independent.
\end{lemma}

\begin{proof}
Condition on all oracle coordinates $u_{j,r}$.  Then $Y_j$ and hence
$X_j$ are fixed.  The coordinate
$v_{j,X_j}$ 
belongs to the disjoint $v$-block and remains an unbiased oracle bit.
Therefore $
\Pr[H_j=b\mid Y_j]=\frac12.$ 
Averaging over $Y_j$ proves the first claim.

For different indices $j$, the entire collections of oracle
coordinates used to define $H_j$ have different lengths and are
therefore disjoint.  Independence follows from the product structure
of the random oracle.
\end{proof}

\begin{lemma}
\label{prop:infinite2}
For each $b\in\bits$,
\begin{equation} 
\Prb_A[L_b^A\text{ is infinite}]=1.
\end{equation} 
Indeed, both $L_0^A$ and $L_1^A$ are infinite simultaneously with
probability one.
\end{lemma}

\begin{proof}
By Lemma~\ref{lem:fair} and the second Borel-Cantelli lemma. 
  
\end{proof}


Fix a deterministic oracle machine $M$ and an index $j$.  Let
\begin{equation}
E_j
=
\{M^A(1^{t_j})\text{ accepts}\}.
\label{7.1}
\end{equation}

For a fixed $b\in\bits$, define the false-positive event
\begin{equation}
B_{j,b}
=
E_j\cap\{H_j\neq b\}.
\label{7.2}
\end{equation}

Finally, define
\begin{equation}
T_j
=
\left\{
M^A(1^{t_j})
\text{ queries the distinguished coordinate }
Z_j(A)
\right\}.
\label{7.3}
\end{equation}

The event $T_j$ is the exceptional case.  If it occurs, the
computation has explicitly discovered an address containing the unique
preimage $X_j$.  If it does not occur, acceptance cannot depend on
the target bit $H_j$.
It is useful to formalize the last observation without informally
speaking about ``conditioning on all bits except a randomly selected
coordinate.''

Define $
\Phi_j:\Omega\longrightarrow\Omega$
by flipping the distinguished coordinate. Formally: 
\begin{equation}
(\Phi_jA)(z)
=
\begin{cases}
1-A(z),&z=Z_j(A),\\
A(z),&z\neq Z_j(A).
\end{cases}
\label{7.4}
\end{equation}

Since $Z_j(A)$ depends only on the $u$-coordinates, which are not
changed by $\Phi_j$, $
Z_j(\Phi_jA)=Z_j(A).$
Hence $
\Phi_j^2=\mathrm{id}.$ 
Moreover, $\Phi_j$ is measure-preserving: after conditioning on the
$u$-block, it simply flips one fixed unbiased coordinate of the
independent $v$-block.

If $A\not\in T_j$ then the computations relative to $A$ and $\Phi_j A$ receive identical answers on any queries, and therefore have the same acceptance outcome. Hence $E_j\Delta \Phi_{j}(E_j)\subseteq T_j$. 

The same coupling shows that $T_j$ itself is $\Phi_j$-invariant. If the distinguished coordinate is never queried relative to $A$, all answers agree in the two computations, so the same transcript is followed and the coordinate is not queried relative to $\Phi_jA$ either. If it is queried, consider the first such query: before it the machine has received no answer from the distinguished coordinate, so the two transcripts agree up to that point and the first distinguished query occurs at the same moment in both. Hence

$$T_j(A)=T_j(\Phi_jA).$$

\begin{lemma}
\label{lem:orbit-structure}
Fix $j\geq2$ and let $\mathcal F$ be a $\sigma$-field generated by
finitely many oracle coordinates, all of length $<t_{j+1}$, that includes
every coordinate on which $E_j$, $T_j$ and $H_j$ depend, and let
\[
  \mathcal G=\bigl\{\,C\in\mathcal F:\Phi_j^{-1}(C)=C\,\bigr\}.
\]
Then:
\begin{itemize}
\item[(i). ] $\Phi_j$ permutes the atoms of $\mathcal F$, and all atoms of
      $\mathcal F$ have the same positive probability;
\item[(ii)] Every $\Phi_j$-orbit of atoms of $\mathcal F$ has
exactly two elements, and the atoms of $\mathcal G$ are precisely
the sets $C\cup\Phi_j(C)$, 
where $C$ ranges over the atoms of $\mathcal F$.
\item[(iii). ] $T_j$ is a union of atoms of $\mathcal G$; in particular
      $T_j\in\mathcal G$.
\end{itemize}
\end{lemma}

\begin{proof}
(i) The distinguished coordinate $Z_j(A)$ has length $t_j<t_{j+1}$ and
depends only on the $u$-coordinates $u_{j,1},\dots,u_{j,m_j}$, which
$\Phi_j$ does not change.  Hence $\Phi_j$ acts on the finitely many
coordinates generating $\mathcal F$ by flipping one of them, determined by
the others, and therefore maps each atom of $\mathcal F$ (a cylinder on
these coordinates) onto another atom.  All these cylinders fix the same
number of coordinates, so they have equal, positive probability under the
uniform product measure.

(ii) An orbit has one or two elements.  If an atom $C$ satisfied
$\Phi_j(C)=C$, then picking $A\in C$ we would get $H_j(\Phi_jA)=H_j(A)$,
because $H_j$ is constant on atoms of $\mathcal F$; but
$H_j(\Phi_jA)=1-H_j(A)$ by the definition~\eqref{7.4} of
$\Phi_j$, a contradiction.  So every orbit has exactly two elements.  
A set belongs to $\mathcal G$ exactly when it is a
$\Phi_j$-invariant union of atoms of $\mathcal F$, equivalently,
when it is a union of sets of the form $C\cup\Phi_j(C)$.  Since
every orbit has two elements, these orbit-unions are precisely the
atoms of $\mathcal G$.

(iii) $T_j$ is $\mathcal F$-measurable and 
$\Phi_j$-invariant; by (ii) it is therefore a union of orbits.
\end{proof}



\begin{lemma}
\label{lem:query-error}
Let $\mathcal F,\mathcal G$ be as in Lemma~\ref{lem:orbit-structure}, and let
\begin{equation} 
  P_j=\Pr(E_j\mid\mathcal G),
  Q_{j,b}=\Pr(B_{j,b}\mid\mathcal G).
\end{equation} 
Then, for every $b\in\{0,1\}$,
\begin{equation} 
  Q_{j,b}
  \geq
  \frac12\bigl(P_j-\mathbf 1_{T_j}\bigr).
  \label{foo-b} 
\end{equation} 
\end{lemma}

\begin{proof}
Fix an atom $D$ of $\mathcal G$. By Lemma~\ref{lem:orbit-structure}(ii), there are
two distinct atoms $C_0,C_1$ of $\mathcal F$ such that $D=C_0\cup C_1$, 
$\Phi_j(C_0)=C_1$, 
$\Phi_j(C_1)=C_0.$
  
By Lemma~\ref{lem:orbit-structure}(i), $C_0$ and $C_1$ have equal positive probability.
Consequently, conditionally on $D$, the choice between the two
$\mathcal F$-atoms is uniform.

We now pass explicitly to the associated two-point quotient.  Let $\widehat\Omega_D=\{0,1\}$
with the uniform probability measure, and define the involution $\widehat\Phi_D(r)=1-r$. 
Since $E_j$, $T_j$, and $H_j$ are $\mathcal F$-measurable, they
are constant on each $C_r$. Define
$ \widehat E_D=
  \{r\in\widehat\Omega_D:C_r\subseteq E_j\},$ 
$ \widehat T_D=\{r\in\widehat\Omega_D:C_r\subseteq T_j\}$, 
and define $\widehat H_D:\widehat\Omega_D\to\{0,1\}$ by $
  \widehat H_D(r)=H_j(A), \text{for any }A\in C_r.$
These definitions do not depend on the choice of $A\in C_r$.

The hypotheses of Lemma~\ref{lem:involution-rare-or-wrong} hold on $\widehat\Omega_D$. Indeed,
from $H_j(\Phi_j A)=1-H_j(A)$
we obtain $\widehat H_D(\widehat\Phi_D(r))
  =1-\widehat H_D(r).$
The $\Phi_j$-invariance of $T_j$ gives
$\widehat\Phi_D(\widehat T_D)=\widehat T_D$, 
and the inclusion $
  E_j\mathbin{\triangle}\Phi_j(E_j)\subseteq T_j$ 
descends to
$\widehat E_D
  \mathbin{\triangle}
  \widehat\Phi_D(\widehat E_D)
  \subseteq
  \widehat T_D.
$

Put $\widehat B_{D,b}=\widehat E_D\cap\{\widehat H_D\neq b\}.$ 

Lemma~\ref{lem:involution-rare-or-wrong}, applied to the two-point probability space
$\widehat\Omega_D$, therefore gives
\begin{equation} 
  \Pr_{\widehat\Omega_D}(\widehat B_{D,b})
  \geq
  \frac12\left(
    \Pr_{\widehat\Omega_D}(\widehat E_D)
    -
    \Pr_{\widehat\Omega_D}(\widehat T_D)
  \right).
\end{equation} 
Because $C_0$ and $C_1$ have equal probability, the quotient
probabilities are exactly the corresponding probabilities
conditioned on $D$. Thus
\begin{equation} 
  \Pr(B_{j,b}\mid D)
  \geq
  \frac12\left(
    \Pr(E_j\mid D)-\Pr(T_j\mid D)
  \right).
\end{equation} 

Finally, $T_j\in\mathcal G$ by Lemma~\ref{lem:orbit-structure}(iii). Hence, on the
$\mathcal G$-atom $D$, $\Pr(T_j\mid D)=\mathbf 1_{T_j}$. Moreover, on $D$, $P_j=\Pr(E_j\mid D)$, 
$Q_{j,b}=\Pr(B_{j,b}\mid D)$.  It follows that
$Q_{j,b}
  \geq
  \frac12\bigl(P_j-\mathbf 1_{T_j}\bigr)$ 
on $D$. Since $D$ was an arbitrary atom of $\mathcal G$, the
inequality holds everywhere.
\end{proof}

We next prove that the exceptional events $T_j$ occur with summable total probability:

\begin{lemma}
\label{lem:invert}
Fix a deterministic polynomial-time oracle machine $M$. Then there is a probabilistic polynomial-time inverter $I_M$ such that,
for every sufficiently large $j$,
\begin{equation}
\Prb_A[T_j]
=
\Prb_{y\leftarrow U_{m_j},r}
\left[
f_{m_j}
\bigl(I_M(1^{m_j},y,r)\bigr)
=
y
\right].
\label{8.1}
\end{equation}
\end{lemma}

\begin{proof}
Given $1^m$ and $y\in\bits^m$, the inverter first sets
$n=4m.$
It tests whether $n=t_j$
for some $j\geq2$.  If not, it fails.

Suppose that indeed $n=t_j$ and hence $m=m_j$.  The inverter simulates $M(1^{n})$
together with a lazy random oracle.
For the programmed coordinates $u_{j,r}$, this lazy oracle responds
\begin{equation} 
A(u_{j,r})=y_r.
\label{8.2}
\end{equation} 
All other oracle queries are answered by independent unbiased bits,
with a table used to maintain consistency on repeated queries. Whenever $M$ asks a query of the form
$v_{j,x}, (x\in\bits^m)$, 
the inverter computes $f_m(x)$.  If
$f_m(x)=y$, 
it outputs $x$.

Because $f_m$ is a permutation, there is exactly one such string,
namely $x=f_m^{-1}(y).$
Thus the inverter succeeds precisely when the simulated machine asks
the distinguished query $v_{j,f_m^{-1}(y)}.$

It remains to compare distributions.  When $y$ is uniform, the
programmed coordinates (\ref{8.2}) are themselves independent unbiased
bits.  Every other oracle coordinate is generated independently and
uniformly by lazy sampling.  Hence the simulated oracle has exactly
the random-oracle distribution.  This proves (\ref{8.1}).

Finally, since $n=4m$, simulating a polynomial-time computation on
input $1^n$ takes time polynomial in $m$.  The lazy oracle table
contains only polynomially many strings of polynomial length.
Therefore $I_M$ is probabilistic polynomial time.
\end{proof}

\begin{corollary}
\label{cor:summableT}
For every polynomial-time oracle machine $M$,
\begin{equation}
\sum_{j=2}^\infty \Prb_A[T_j]<\infty.
\label{8.3}
\end{equation}
Consequently,
\begin{equation}
\Prb_A[T_j\text{ infinitely often}]=0.
\label{8.4}
\end{equation}
\label{cor-cantelli}
\end{corollary}

\begin{proof}
By one-wayness and Lemma~\ref{lem:invert}, there is a negligible 
function $\nu_M$ such that
$\Prb_A[T_j]\leq\nu_M(m_j).$ 
Since $\nu_M$ is negligible, eventually $\nu_{M}(k)\leq \frac{1}{k^2}$. The convergence of $\zeta(2)$ completes the proof. 
 
The first Borel-Cantelli lemma gives (\ref{8.4}). 
\end{proof}

\begin{remark} 
Corollary~\ref{cor-cantelli} is the precise point at which cryptographic average-case
one-wayness is used.  A purely worst-case noninvertibility assumption
would not imply (\ref{8.3}).
\end{remark}

Fix now a clocked polynomial-time oracle machine $M$, and let $q_M$ be an increasing polynomial $q_M(n)\geq n$ bound
both its running time and the lengths of its oracle queries.  

\begin{lemma}
\label{prop:fixedM}
Fix $b\in\bits$.  With probability one over $A$,
\begin{equation}
L(M^A)\subseteq L_b^A
\quad\Longrightarrow\quad
L(M^A)\text{ is finite}.
\label{11.1}
\end{equation}
\end{lemma}
\begin{proof} 
As in previous proofs, discard
a finite initial range so that $q_M(t_j)<t_{j+1}.$

Define
\begin{equation}
\mathcal F_j
=
\sigma\bigl(
A(z):|z|<t_{j+1}
\bigr).
\label{9.2}
\end{equation}
Then the acceptance event $E_j$ belongs to $\mathcal F_j$.

Let $\Phi_j$ be the involution from (\ref{7.4}), and define
\begin{equation}
\mathcal G_j
=
\left\{
C\in\mathcal F_j:
\Phi_j^{-1}(C)=C
\right\}.
\end{equation}

Since $q_M(t_j)<t_{j+1}$ and the coordinates $u_{j,r},v_{j,x}$ have length $t_j$, the pair $(\mathcal F_j,\mathcal G_j)$ satisfies the hypotheses of Lemma~\ref{lem:orbit-structure}.

Since $\Phi_j$ changes only an oracle bit of length exactly $t_j$, $
\mathcal F_{j-1}
\subseteq
\mathcal G_j
\subseteq
\mathcal F_j.$ 

Likewise, $
\mathcal F_j\subseteq\mathcal G_{j+1}$, 
because $\Phi_{j+1}$ changes a coordinate of length $t_{j+1}$, while
$\mathcal F_j$ contains only coordinates of smaller length. Hence 
\begin{equation} 
\mathcal F_{j-1}
\subseteq
\mathcal G_j
\subseteq
\mathcal F_j
\subseteq
\mathcal G_{j+1}
\subseteq\cdots
\label{filt-1} 
\end{equation} 
is an increasing interleaved filtration.


We now apply Lemma~\ref{lem:CBC} to filtration~(\ref{filt-1}), inserting the empty event at odd
indices.  Thus, if
\begin{equation}
P_j=\Pr(E_j\mid\mathcal G_j),
\label{10.2}
\end{equation}
then almost surely
\begin{equation}
E_j\text{ i.o.}
\quad\Longleftrightarrow\quad
\sum_jP_j=\infty.
\label{10.3}
\end{equation}

Similarly, putting
\begin{equation}
Q_{j,b}
=
\Pr(B_{j,b}\mid\mathcal G_j),
\label{10.4}
\end{equation}
we obtain almost surely
\begin{equation}
B_{j,b}\text{ i.o.}
\quad\Longleftrightarrow\quad
\sum_jQ_{j,b}=\infty.
\label{10.5}
\end{equation}

We now complete the proof of Lemma~\ref{prop:fixedM}. 
Take the intersection of the following probability-one events:

\begin{itemize}
\item[(*). ] by Corollary~\ref{cor:summableT}, only finitely many $T_j$
occur;
\item[(**). ] the equivalence (\ref{10.3}) holds;
\item[(***). ] the equivalence (\ref{10.5}) holds.
\end{itemize}

Fix an oracle $A$ in this intersection.  We now regard
$P_j(A)$ and $Q_{j,b}(A)$ as ordinary nonnegative real numbers.

There are two cases: 

\begin{itemize} 
\item{\textbf{Case 1:} $
\sum_jP_j(A)<\infty.$} 

By conditional Borel-Cantelli, only finitely many events $E_j$
occur.  Hence $M^A$ accepts only finitely many support strings
$1^{t_j}$. If $
L(M^A)\subseteq L_b^A$ 
then $M^A$ accepts no string outside $T$, because $L_b^A\subseteq T$.
Thus $L(M^A)$ is finite.

\medskip
\noindent
\item{\textbf{Case 2:} $\sum_jP_j(A)=\infty.$} 

Because of (*), $T_j(A)=0$ for all sufficiently large $j$. The query-or-error
inequality (\ref{foo-b}) gives eventually
\begin{equation} 
Q_{j,b}(A)
\geq
\frac12P_j(A).
\end{equation} 
Therefore
\begin{equation} 
\sum_jQ_{j,b}(A)=\infty.
\end{equation} 
Conditional Borel-Cantelli Lemma~\ref{lem:CBC} implies that $B_{j,b}$ occurs infinitely
often.  In particular, for some $j$,
$M^A(1^{t_j})\text{ accepts}$ 
but $H_j(A)\neq b.$ 
Hence $
1^{t_j}\in
L(M^A)\setminus L_b^A$, 
so $
L(M^A)\not\subseteq L_b^A.$ 
\end{itemize} 

Thus, whenever $L(M^A)\subseteq L_b^A$, only Case~1 is possible, and
then $L(M^A)$ is finite.

\end{proof} 

Intersecting the probability 1 events in Lemma~\ref{prop:fixedM} over the countably many clocked machines $M$ and both values of $b\in \bits$, and intersecting once more with the event from Lemma~\ref{prop:infinite2} yields a measure-one set of oracles such that both $L_0^A$ and $L_{1}^{A}$ are infinite and have no infinite subset in $\Pclass^A$,  completing the proof of Theorem~\ref{thm:npconp}.

\section{Proof of the Main Theorem}
\label{sec:proofmain}

We now prove Theorem~\ref{thm:main}.  Our approach does not use directly the Sipser functions from \cite{1}, but a family $G_{d,m}$ based on a minor reparametrization $\widehat{Sipser}_{d,m}$ of the RST family. More precisely 

\begin{equation}\label{eq:effective-orientation}
  G_{d,m}=
  \begin{cases}
    \widehat{Sipser}_{d,m},&d\text{ even},\\
    1-\widehat{Sipser}_{d,m},&d\text{ odd}.
  \end{cases}
\end{equation}

The reason for this choice is technical: we want that the boolean circuit representing our test functions has an OR gate at the top. In the second case, pushing the complement to the leaves exchanges AND
and OR and negates every input literal. 

As for the modified Sipser functions $\widehat{Sipser}_{d,m}$ used in the definition of $G_{d,m}$: 
\begin{itemize} 
\item The depth, alternating AND/OR structure, read-once property, root orientation, bottom fan-in \(m\) and intermediate fan-in $w$ are the same as for the family $Sipser_{d,m}$ in \cite{1}. 
\item RST’s top fan-in \(w_0\), defined indirectly through their restriction parameters, is replaced by a new fan-in $\widehat w_0$ that is computable in polynomial time by elementary rational arithmetic, and that agrees with $w_0$ up to a \emph{relative} error $2^{-\Omega_d(m)}$. The net result will be a class of formulas $\widehat{Sipser}_{d,m}$, that will, in general have a different number of variables than
$Sipser_{d,m}$, so a direct comparison between the two formulas is not possible. However, we will show that there is a way to couple their input variables, with uniform marginals, such that 
\begin{equation} 
\Prb_{X,Y,Z}[\widehat{Sipser}_{d,m}(X,Y)\neq Sipser_{d,m}(X,Z)]\leq 2^{-\Omega_{d}(m)}
\end{equation} 

This will allow us to transfer the correlation upper bound from \cite{1} (which is valid for the family $Sipser_{d,m}$) to the new class of functions $\widehat{Sipser}_{d,m}$ and, ultimately, to functions $G_{d,m}$.  

\item The reason we use the modified version of the Sipser function is that we need to ensure the polynomial time computability of the fan-ins, in order to prove that the test language we construct actually belongs to the claimed level of the relativized polynomial hierarchy. For the correlation bounds in \cite{1} ignoring this computability was harmless, but in our case it may affect the complexity of the constructed test language. Whereas for the intermediate fan-ins used in \cite{1} we could prove that they are computable in polynomial time \emph{exactly}, for the top fan-in the situation is more delicate, and motivated the use of a polynomial time computable approximation $\widehat{w_0}$ to $w_0$. It is worth stressing at the outset why, in our proof strategy,  an \emph{exact} value is indispensable for the intermediate fan-in $w$ but a relative approximation suffices for the top fan-in: the coupling of Lemma~\ref{lem:coupling} matches root subtrees \emph{together with all the leaves below them}, so a discrepancy in $w$ (or in the bottom fan-in $m$) would perturb every one of the $\approx 2^m$ root subtrees simultaneously, whereas a discrepancy $\Delta$ in the top fan-in leaves $\Delta$ subtrees unmatched, each of which contributes at most $p_{d-1}\leq 2^{1-m}$ to the disagreement probability\footnote{One can, indeed, use approximate fan-ins at all levels, at the expense of a more complicated coupling strategy.}. 
\end{itemize} 


\subsection{Polynomial time computability of the intermediate fan-ins.} 

We recall the following standard number-theoretic input (with a deliberately weak exponent, since that's all that we need): 

\begin{lemma}
\label{lem:effective-ln2}
There exist effectively computable constants $c>0$ and $Q_0\geq 1$
such that, for all integers $P$ and all integers $Q\geq Q_0$,
\[
   \left|\ln 2-\frac{P}{Q}\right|\geq c\cdot Q^{-5}.
\]
\end{lemma}

\begin{proof} One can take $c=(2000)^{-1}$ and $Q_0=1$.  
See \cite{51}, Corollary 2 of Theorem 3, page 150. 
\end{proof}

\begin{lemma}
\label{lem:compute-w}
The integer
$w(m)=\left\lfloor m2^m\ln 2\right\rfloor$ 
is computable in deterministic time polynomial in $m$.
\end{lemma}

\begin{proof}
Put $q_m=m\cdot 2^m$.  Lemma~\ref{lem:effective-ln2} gives, for every
integer $P$ and all sufficiently large $m$,
\begin{equation} 
   |q_m\ln 2-P|\geq c\cdot q_m^{-4}.
\end{equation} 
Since $q_m\leq 2^{2m}$ for $m\geq2$, there are effectively
computable constants $C>0$ and $m_0$ such that
$\operatorname{dist}(q_m\ln 2,\mathbb Z)
      \geq 2^{-Cm}
      \mbox{ for }(m\geq m_0)$. 

We use the positive-term expansion
\begin{equation} 
   \ln 2=\sum_{r=1}^{\infty}\frac{1}{r2^r}.
\end{equation} 
For $s\geq1$, define
\begin{equation} 
   \ell_s=\sum_{r=1}^{s}\frac{1}{r2^r},
   \qquad
   u_s=\ell_s+\frac{1}{(s+1)2^s}.
\end{equation} 
The tail satisfies
\begin{equation} 
  0<\ln2-\ell_s
    =\sum_{r=s+1}^{\infty}\frac{1}{r2^r}
    \leq
      \frac{1}{s+1}\sum_{r=s+1}^{\infty}2^{-r}
    =\frac{1}{(s+1)2^s},
\end{equation} 
and hence $\ell_s<\ln2\leq u_s$.

Take $s=\lceil(C+3)m\rceil$.  Then
\begin{equation} 
  q_m(u_s-\ell_s)
     \leq q_m\cdot 2^{-s}
     \leq 2^{2m-(C+3)m}
     <2^{-Cm}.
\end{equation} 
Thus the rational interval $ I_m=[q_m\ell_s,q_mu_s]$
contains $q_m\ln2$ and has length strictly smaller than
$\operatorname{dist}(q_m\ln2,\mathbb Z)$.  It follows that
$I_m$ contains no integer.  Consequently,
$\lfloor q_m\ell_s\rfloor
   =\lfloor q_m u_s\rfloor
   =\lfloor q_m\ln2\rfloor=w(m).$

The rationals $\ell_s$ and $u_s$ can be formed exactly using
$s=O(m)$ additions and multiplications of integers of polynomially
many bits.  Even the naive common-denominator computation uses only
$O(s^2)$ bits, and therefore runs in time polynomial in $m$.
The finitely many cases $m<m_0$ may be hardwired.
\end{proof}

For the top fan-in $w_0=\lceil R_m \rceil$ (see equation~(\ref{eq:w0-definition})) no analogue of Lemma~\ref{lem:compute-w} is available, and none is obviously within reach: unlike $w$, the number $R_m$ is not an explicit closed-form expression but the output of a chain of $\Theta(d)$ transcendental operations, so an exact computation would require a separation-from-the-integers statement for $R_m$ of the kind that Lemma~\ref{lem:effective-ln2} provides for $\ln2$.

Fortunately we do not need one.  As explained in the third bullet above, the top fan-in enters the argument only through the number of unmatched root subtrees in the coupling of Lemma~\ref{lem:coupling}, and each unmatched subtree contributes only $p_{d-1}\leq2^{1-m}$ to the disagreement probability.  It therefore suffices to approximate $R_m$ to \emph{relative} accuracy $2^{-\Omega_d(m)}$.  That is an easy task, because the leading term of $R_m$ is the single explicit constant $2^m\ln2$:

\begin{lemma}
\label{lem:R-asymptotic-lemma}
For every fixed $d\geq2$ and all sufficiently large $m$,
\begin{equation}
       \frac q2\leq t_r\leq2q
       \qquad(1\leq r\leq d-1),
  \label{eq:t-range-main}
\end{equation}
\begin{equation}
       2^{-m-1}\leq D_m\leq2^{2-m},
  \label{eq:D-range}
\end{equation}
and
\begin{equation}
  R_m=2^m\ln2\left(1+O_d\!\left(2^{-m/20}\right)\right).
  \label{eq:R-asymptotic}
\end{equation}
\end{lemma}

\begin{proof}
Recall that
\begin{equation}
 \rho=2^{-m},\qquad q=2^{-m/2},\qquad
 w=\lfloor m2^m\ln2\rfloor,
 \qquad
 \lambda=\frac{(\log_2 w)^{3/2}}{w^{5/4}},
\end{equation}
and
\begin{equation}
 t_{d-1}=\frac{\rho-\lambda}{q},
 \qquad
 t_{r-1}=\frac{(1-t_r)^{qw}-\lambda}{q}
       \quad(2\leq r\leq d-1),
\end{equation}

By~\cite[Lemma~7.1]{1} (whose proof is given in
\cite[Appendix~A]{1}) there is a universal constant
$c>0$ such that
\begin{equation}
  |t_r-q|\le q^{1.1}
  \qquad (1\le r\le d-1)
  \label{eq:RST-t-estimate}
\end{equation}
whenever $2\le d\le cm/\log m$.  Since $d$ is fixed, this condition
holds for all sufficiently large $m$.  Furthermore,
$q^{0.1}=2^{-m/20}\le 1/2$ for $m\ge20$.  Consequently,
\[
  t_r\ge q-q^{1.1}=q(1-q^{0.1})\ge\frac q2
  \qquad\text{and}\qquad
  t_r\le q+q^{1.1}=q(1+q^{0.1})
       \le\frac{3q}{2}\le2q,
\]
which is~\eqref{eq:t-range-main}; in particular, after increasing
$m_0=m_0(d)$ if necessary, $2q\leq1/2$.

The inequalities
\begin{equation}
       x\leq-\ln(1-x)\leq\frac{x}{1-x}
       \qquad(0\leq x<1)
  \label{eq:log-sandwich}
\end{equation}
applied with $x=t_1$ therefore give
\begin{equation}
   \frac12q^2
       \leq D_m
       \leq\frac{2q^2}{1-2q}
       \leq4q^2,
\end{equation}
that is, \eqref{eq:D-range}.  For the sharper
statement~\eqref{eq:R-asymptotic}, write
$-\ln(1-t_1)=t_1(1+\xi)$ with, by~\eqref{eq:log-sandwich},
$0\leq\xi\leq t_1/(1-t_1)\leq4q$.  Since
$t_1=q\bigl(1+O(q^{0.1})\bigr)$ by~\eqref{eq:RST-t-estimate} and
$q\leq q^{0.1}=2^{-m/20}$, we get
\[
  D_m=q\,t_1(1+\xi)
     =q^2\bigl(1+O(2^{-m/20})\bigr)
     =\rho\bigl(1+O(2^{-m/20})\bigr),
\]
and \eqref{eq:R-asymptotic} follows from $R_m=\ln2/D_m$ and
$\rho^{-1}=2^m$.
\end{proof}

Note that~\eqref{eq:R-asymptotic} carries an absolute error
$O(2^{19m/20})$ and therefore says nothing whatsoever about the integer
$w_0$ itself.  What it does deliver is precisely the relative accuracy
that the coupling can absorb, and that is what we exploit in the next
subsection.

\subsection{The reparameterized Sipser functions.}

For integers $m\geq 2$, $d\geq2$, let $\ell_m$ denote the $2m$-th partial sum
of the series for $\ln2$ already used in the proof of
Lemma~\ref{lem:compute-w}, and let $\widehat w_0$ be the integer it
determines:
\begin{equation}
  \widehat{\ell_m}=\ell_{2m}=\sum_{r=1}^{2m}\frac{1}{r2^{r}},
  \qquad
  \widehat w_0=\left\lfloor 2^{m}\widehat{\ell_m}\right\rfloor .
  \label{eq:modified-w0}
\end{equation}
By the tail estimate already recorded in the proof of
Lemma~\ref{lem:compute-w} (with $s=2m$),
\begin{equation}
  0<\ln2-\widehat{\ell_m}\leq\frac{1}{(2m+1)2^{2m}}<2^{-2m}.
  \label{eq:ell-tail}
\end{equation}

Let $\widehat{Sipser}_{d,m}$ be the regular, read-once, monotone
formula with $d$ alternating layers, bottom AND gates, and fan-ins
\begin{equation}
  \widehat{w_0},\underbrace{w,\ldots,w}_{d-2\text{ intermediate levels}}, m.
  \label{eq:effective-fanin-vector}
\end{equation}

The intermediate range is, of course, empty when $d=2$. Thus the root is OR when
$d$ is even and AND when $d$ is odd. Write
\begin{equation}\label{eq:effective-N}
  N(d,m)=m\cdot \widehat{w_0}\cdot w^{d-2}.
\end{equation}

\subsection{Useful properties of the reparameterized Sipser functions.} 

Next, we record some useful properties of the family $\widehat{Sipser}_{d,m}$:  

First of all, the change in top-width is, as promised, small enough to make the upper bound work: 

\begin{lemma}
\label{lem:w0-distance}
The integer $\widehat w_0$ is computable in deterministic time polynomial
in $m$, and $\tfrac12<\widehat{\ell_m}<\ln2<1$.  Moreover, for every fixed $d\geq 2$ there exists a constant $C_d>0$ such that, for all sufficiently
large $m$,
\begin{equation}
  |\widehat w_0-w_0|\leq\Delta_m:=C_d\cdot 2^{\frac{19}{20}m}.
  \label{eq:w0-distance}
\end{equation}
\end{lemma}

\begin{proof}
\emph{Computability.}  The rational $\widehat{\ell_m}$ has reduced denominator
dividing $\operatorname{lcm}(1,\dots,2m)\cdot2^{2m}$, a number of $O(m)$
bits by the elementary bound $\operatorname{lcm}(1,\dots,n)\leq4^{n}$.
Summing the $2m$ terms over this common denominator and performing a
single integer division computes $\lfloor2^m\widehat{\ell_m}\rfloor$ exactly, using
$O(m)$ arithmetic operations on $O(m)$-bit integers.

\emph{Size.}  The first two terms already give
$\widehat{\ell_m}\geq\frac12+\frac18=\frac58>\frac12$, and $\widehat{\ell_m}<\ln2<1$
by~\eqref{eq:ell-tail}.

\emph{Accuracy.}  By~\eqref{eq:ell-tail},
$|2^m\widehat{\ell_m}-2^m\ln2|\leq2^m\cdot2^{-2m}=2^{-m}\leq1$, and taking the
floor costs at most $1$ further, so $|\widehat w_0-2^m\ln2|\leq2$.
By~\eqref{eq:R-asymptotic}, $|R_m-2^m\ln2|\leq2^{m}\ln2\cdot
O_d(2^{-m/20})$, so there exists $K_{d}>0$ such that 
for all sufficiently large $m$, 
$|R_m-2^m\ln2| \leq K_{d}\cdot 2^{\frac{19}{20}m}$. 
On the other hand $|w_0-R_m|\leq1$, because $w_0=\lceil R_m\rceil$.  Adding the three
estimates gives
$|\widehat w_0-w_0|\leq3+K_{d}\cdot 2^{\frac{19}{20}m}\leq C_{d}\cdot 2^{\frac{19}{20}m}$, 
for all sufficiently large $m$, where e.g. $C_{d}=K_{d}+3$. 
\end{proof}

\begin{lemma} 
\label{lemma-validity} 
$\widehat w_0$ and all the other fan-ins of
$\widehat\Sipser_{d,m}$, $N(d,m)$, validity of a child index, the rank of
a leaf, and the rank of a valid root-to-leaf path are computable in time polynomial in $m$ (that may depend on $d$). 

Finally, for every fixed $d\geq2$ and all sufficiently large $m$,
\begin{equation}
\label{eq:effective-size}
  2^{-(d-1)}m^{d-1}2^{(d-1)m}
  \leq N(d,m)
  \leq m^{d-1}2^{(d-1)m}.
\end{equation}
Consequently,
\begin{equation}
\label{eq:N-coarse-bounds}
  2^m\leq N(d,m)\leq 2^{2dm}
\end{equation}
for all sufficiently large $m$.
\end{lemma} 
\begin{proof} 
Polynomial-time computability of fan-ins follows from
Lemma~\ref{lem:compute-w}, Lemma~\ref{lem:w0-distance},  and the
fact that the bottom fan-in is $m$.

To index bottom leaves we use zero-based child indices. Denote by $f_r$ the fan-in of the $r$-th level from the top, according to equation~(\ref{eq:effective-fanin-vector}). A path $(i_0,\ldots,i_{d-1})$ from the root to the given leaf is valid
exactly when $0\leq i_r<f_r$ for every $r$, and its rank is
\begin{equation}\label{eq:effective-rank}
  \operatorname{rank}_m(i_0,\ldots,i_{d-1})
  =\sum_{r=0}^{d-1} i_r\prod_{s=r+1}^{d-1}f_s.
\end{equation}
This is a bijection from the valid paths onto
$\{0,\ldots,N(d,m)-1\}$. Since $d$ is fixed, its evaluation and all
validity tests use a fixed number of arithmetic operations on
$O_d(m)$-bit integers. 

Directly from~\eqref{eq:modified-w0} and $\tfrac12<\ell_m<1$
(Lemma~\ref{lem:w0-distance}) we obtain, for every $m\geq1$,
\begin{equation}
       2^{m-1}\leq\widehat w_0\leq2^m,
  \label{eq:w0-hat-size}
\end{equation}
since $2^m\ell_m>2^{m-1}$ forces
$\lfloor2^m\ell_m\rfloor\geq2^{m-1}$ (as $2^{m-1}$ is an integer), while
$\widehat w_0\leq2^m\ell_m<2^m$.
Moreover, $w=\lfloor m\cdot 2^m\ln2\rfloor$, 
and therefore, again using $1/2<\ln2<1$,
\begin{equation} 
       \frac12\cdot m\cdot 2^m\leq w\leq m\cdot 2^m
\end{equation}
for all sufficiently large $m$.  Recalling that
$N(d,m)=m\cdot \widehat{w_0}\cdot w^{d-2}$, 
we obtain
\begin{equation} 
 N(d,m)\geq
 m\,2^{m-1}
 \left(\frac12m2^m\right)^{d-2}
 =2^{-(d-1)}m^{d-1}2^{(d-1)m},
\end{equation} 
\begin{equation} 
 N(d,m) \leq
 m\,2^m(m2^m)^{d-2}
 =m^{d-1}2^{(d-1)m}.
\end{equation} 
This proves \eqref{eq:effective-size}.  In particular, $
  N(d,m)=\Theta_d\!\left(m^{d-1}2^{(d-1)m}\right)$, 
$\log N(d,m)=\Theta_d(m)$, and all fan-ins and $N(d,m)$ have $O_d(m)$ bits. 
\end{proof} 

\begin{lemma}
\label{lem:input-counts}
For fixed $d$,
\begin{equation}
  \frac{N(d,m)}{N_{RST}(d,m)}=1+O_{d}\!\left(2^{-m/20}\right).
  \label{eq:input-count-comparison}
\end{equation}
In particular $N(d,m)\leq2N_{RST}(d,m)$ and $w_0\geq2^{m-2}$ for all
sufficiently large $m$.
\end{lemma}

\begin{proof}
The two formulas differ only in their top fan-in, so
\[
  \frac{N(d,m)}{N_{RST}(d,m)}=\frac{\widehat w_0}{w_0}.
\]
By~\eqref{eq:w0-hat-size} and Lemma~\ref{lem:w0-distance},
$w_0\geq\widehat w_0-\Delta_m\geq2^{m-1}-C_{d}\cdot 2^{\frac{19}{20}m}\geq2^{m-2}$
for all sufficiently large $m$.  Hence
\[
  \left|\frac{\widehat w_0}{w_0}-1\right|
  \leq\frac{\Delta_m}{w_0}\leq \frac{C_{d}\cdot 2^{19m/20}}{2^{m-2}}=4C_{d}\cdot 2^{-m/20}. 
\]
which is~\eqref{eq:input-count-comparison}.
\end{proof}

The most important result is that the reparameterized class of formulas can be coupled with the original one: 

\begin{lemma} 
There is a way to couple (some of) the variables in the two formulas so that, if $X$ denotes the variables in common to both formulas, $Y$ the set of variables private to the first formula and $Z$ the set of variables private to the second formula 
\begin{equation} 
\Prb_{X,Y,Z}[\widehat{Sipser}_{d,m}(X,Y)\neq Sipser_{d,m}(X,Z)]\leq 2^{-\Omega_{d}(m)}
\label{foofoo}
\end{equation} 
\label{lem:coupling} 
\end{lemma} 
\begin{proof} Define the following coupling: 
\begin{definition}
Put $r=\min(w_0,\widehat w_0)$.  For each $1\leq i\leq r$, identify the
$i$th root child of $\Sipser_{d,m}$ with the $i$th root child of
$\widehat\Sipser_{d,m}$, including all leaves below it.  Assign the same
independent unbiased bits to the corresponding leaves.  The
$|\widehat w_0-w_0|$ remaining root children of whichever formula has the
larger top fan-in are left \emph{unmatched}; assign independent unbiased
bits to every leaf in each unmatched subtree.
\label{def:root-coupling}
\end{definition}

In the following result we use the following conventions: an input to a gate \emph{activates its parent} if the gate is an OR and the input is one or if the gate is an AND and the input is zero. 

\begin{claim}
\label{lem:root-child-rarity} Let $p_h$ be the probability that a height-$h$ non-root RST subformula takes
the value activating its parent. 
For fixed $d$,
\begin{equation}
  p_h=\rho\left(1+O_d(m^{2d}\rho)\right)
  \qquad(1\leq h\leq d-1).
  \label{eq:p-estimate}
\end{equation}
In particular,
\begin{equation}
  p_{d-1}\leq2\rho
  \label{eq:p-upper}
\end{equation}
for all sufficiently large $m$.
\end{claim}

\begin{proof} 

We prove the following stronger assertion by induction on $h$: for
every fixed $h\geq1$ there is a constant $C_h>0$ such that, for all
sufficiently large $m$,
\begin{equation}
  p_h=\rho(1+\varepsilon_h),
  \qquad
  |\varepsilon_h|\leq C_h m^{h-1}\rho.
  \label{eq:strong-ph-estimate}
\end{equation}
The constants and the threshold for $m$ may depend on $h$.  Since
only $1\leq h\leq d-1$ are considered and $d$ is fixed, they may
eventually be chosen uniformly in terms of $d$.

For $h=1$ we have $p_1=2^{-m}=\rho$, 
so~\eqref{eq:strong-ph-estimate} holds with $\varepsilon_1=0$.

Now let $2\leq h\leq d-1$ and suppose that
$p_{h-1}=\rho(1+\varepsilon_{h-1}),
  |\varepsilon_{h-1}|
     \leq C_{h-1}m^{h-2}\rho.$ 
Since $m^{h-2}\rho=m^{h-2}2^{-m}=o(1)$, after increasing the
threshold for $m$ we have $|\varepsilon_{h-1}|\leq1/2$ and hence
$0\leq p_{h-1}\leq\frac32\rho\leq2\rho.$ 
Thus $p_{h-1}\leq1/2$ for sufficiently large $m$, and we may use $\ln(1-x)=-x+O(x^2), (0\leq x\leq1/2)$, 
with an absolute implied constant.

Write $w=m(\ln2)\rho^{-1}+\theta_m, |\theta_m|\leq1,$
which follows from
$w=\lfloor m2^m\ln2\rfloor$ and $\rho=2^{-m}$.  Using the recursion
$p_h=(1-p_{h-1})^w$, we obtain
\begin{align*}
 & \ln\frac{p_h}{\rho}  =w\ln(1-p_{h-1})+m\ln2
   =-wp_{h-1}+m\ln2+O(wp_{h-1}^2) = \\ 
 &  =-m(\ln2)\varepsilon_{h-1}
     -\theta_m\rho(1+\varepsilon_{h-1})
     +O(w\rho^2).
\end{align*}
Here we used $p_{h-1}=O(\rho)$.  Moreover,
$w\rho^2
   =\bigl(m(\ln2)\rho^{-1}+O(1)\bigr)\rho^2
   =O(m\rho).$ 
Consequently,
\begin{align*}
 \ln\frac{p_h}{\rho}
   &=O_h\bigl(m|\varepsilon_{h-1}|\bigr)+O(m\rho)=O_h\bigl(m^{h-1}\rho\bigr).
\end{align*}
Because $m^{h-1}\rho=o(1)$, exponentiation and
$e^x=1+O(x)$ for $x=o(1)$ give
\[
  p_h
   =\rho\cdot \exp\bigl(O_h(m^{h-1}\rho)\bigr)
   =\rho\bigl(1+O_h(m^{h-1}\rho)\bigr).
\]
This proves~\eqref{eq:strong-ph-estimate} and completes the induction.

Finally, since $1\leq h\leq d-1$ and $d$ is fixed, the finitely many
implied constants may be replaced by a single constant depending
only on $d$.  Furthermore, $
  m^{h-1}\rho\leq m^{d-2}\rho\leq m^{2d}\rho$. 
Therefore
\[
  p_h=\rho\bigl(1+O_d(m^{2d}\rho)\bigr),
  \qquad 1\leq h\leq d-1,
\]
which proves~(\ref{eq:p-estimate}).  Since $m^{2d}\rho=m^{2d}2^{-m}=o(1)$, it also
follows that $p_{d-1}\leq2\rho$ for all sufficiently large $m$.
\end{proof}

Now we're in position to prove the lemma: given the coupling described in
Definition~\ref{def:root-coupling}, every matched root subtree agrees
\emph{identically} with its coupled match, not merely with high
probability.  Hence the two root outputs can differ only if some
\emph{unmatched} root child takes the value activating the root.  By
Lemma~\ref{lem:w0-distance} there are at most
$\Delta_m=C_{d}\cdot 2^{\frac{19}{20}m}$ unmatched children, and each takes the
activating value with probability $p_{d-1}\leq2\rho=2^{1-m}$ by
Claim~\ref{lem:root-child-rarity}.  A union bound (no independence is
needed) therefore gives
\begin{equation}
  \Prb_{X,Y,Z}[\widehat{Sipser}_{d,m}(X,Y)\neq Sipser_{d,m}(X,Z)]
  \leq\Delta_m\cdot2^{1-m}
  =2C_{d}\cdot 2^{-m/20},
  \label{eq:coupling-quantitative}
\end{equation}
which proves~\eqref{foofoo}.
\end{proof}

\subsection{Hardness result for family $G_{d,m}$.} 
The following is the specific hardness result we employ:  

\begin{lemma}[Hardness of the uniform family $G_{d,m}$]
\label{lem:effective-sipser}
For every fixed $d\geq2$, there are constants $\alpha_d,\beta_d>0$ such that, for
all sufficiently large $m$, every normal form circuit
\[
  C:\bits^{N(d,m)}\longrightarrow\bits
\]
with an AND root, depth at most $d$ and 
$
  \operatorname{size}(C)\leq2^{N(d,m)^{\alpha_d}}$ 
satisfies
\begin{equation}
  \Prb[C(U)=G_{d,m}(U)]
  \leq\frac12+ N(d,m)^{-\beta_d}.
  \label{eq:G-correlation}
\end{equation}
Here $U=(X,Y)$ is a random input in $\bits^{N(d,m)}$, $X$ consists of the first $min(N(d,m),N_{RST}(d,m))$ bits, and $Y$ of the remaining (if any) bits. 

The dual statement holds for $1-G_{d,m}$ against OR-root circuits.
\end{lemma}

\begin{proof}
We first treat the case where the top gate is an OR, 
that is $G_{d,m}(U)=\widehat{Sipser}_{d,m}(U)$. 

Couple the exact and modified formulas $Sipser_{d,m}(X,Z)$ and $\widehat{Sipser}_{d,m}(X,Y)$ as in
Definition~\ref{def:root-coupling}.  Denote their outputs by $S$ and
$\widehat S$.

On the other hand, circuit $C(X,Y)$ and the formula $Sipser_{d,m}(X,Z)$ may as well have input
sets of slightly different sizes, $N(d,m)$ variables for $C$, $N_{RST}(d,m)$ for $Sipser_{d,m}$.  We will couple them as well, as follows: 

\begin{claim}
\label{claim:C-vs-RST}
We can couple the circuit $C(X,Y)$ and the formula $Sipser_{d,m}(X,Z)$ so that
\begin{equation}
  \Prb[C(X,Y)=Sipser_{d,m}(X,Z)]\leq\frac12+N_{RST}(d,m)^{-b_d},
  \label{eq:RST-after-fixing}
\end{equation}
where $b_d>0$ is supplied by Proposition~\ref{thm:RST}. 
\end{claim}
\begin{proof}
There are two cases:
\begin{itemize}
\item[-] If the modified formula $\widehat{Sipser}_{d,m}(X,Y)$ has the unmatched root subtrees (i.e. $Y\neq \emptyset$, $Z=\emptyset$) fix all
variables in $Y$.  For each fixing $u$, the restricted circuit
$C_u$ has no larger depth or size.  $C_u$ is in normal form (Definition~\ref{def:normal-form}), possibly with constant leaves. Thus the bound from Corollary~\ref{cor:rst} applies to every $C_u$.  Averaging over the independent uniform
choice of $u$ yields~\eqref{eq:RST-after-fixing}.
\item[-] The same reasoning applies
when the exact formula $Sipser_{d,m}(X,Z)$, rather than the modified formula $\widehat{Sipser}_{d,m}(X,Y)$, has the unmatched
subtrees, that is $Z\neq \emptyset$ and $Y=\emptyset$: in that case $C(X,Y)$ simply ignores the variables in $Z$, and we apply again Corollary~\ref{cor:rst}.
\end{itemize}
\end{proof}

Now triangle inequality gives:
\begin{equation}
  \mathbf 1_{\{C(X,Y)=\widehat{Sipser}_{d,m}(X,Y)\}}
  \leq
  \mathbf 1_{\{C(X,Y)=Sipser_{d,m}(X,Z)\}}+
  \mathbf 1_{\{Sipser_{d,m}(X,Z)\neq\widehat{Sipser}_{d,m}(X,Y)\}},
\end{equation}
so Lemma~\ref{lem:coupling}, in the quantitative
form~\eqref{eq:coupling-quantitative}, gives
\begin{equation}
  \Prb[C(U)=\widehat{Sipser}_{d,m}(U)]
  \leq\Prb[C(X,Y)=Sipser_{d,m}(X,Z)]+2C_{d}\cdot 2^{-m/20}.
  \label{eq:triangle-transfer}
\end{equation}

By Lemma~\ref{lem:input-counts} we may take $\alpha_d:=a_d/2$: indeed
$N(d,m)\leq2N_{RST}(d,m)$ for all sufficiently large $m$, so
$N(d,m)^{\alpha_d}\leq(2N_{RST}(d,m))^{a_d/2}\leq N_{RST}(d,m)^{a_d}$ and
hence $2^{N(d,m)^{\alpha_d}}\leq2^{N_{RST}(d,m)^{a_d}}$.  This justifies the
appeal to Corollary~\ref{cor:rst} in
Claim~\ref{claim:C-vs-RST}.  Combining~\eqref{eq:RST-after-fixing} with
\eqref{eq:triangle-transfer} gives
\begin{equation}
  \Prb[C(U)=\widehat{Sipser}_{d,m}(U)]
  \leq\frac12+N_{RST}(d,m)^{-b_d}+2C_{d}\cdot 2^{-m/20}.
  \label{eq:pre-final-bound}
\end{equation}
It remains to absorb both error terms into a single power of $N(d,m)$.
Set
\begin{equation}
  \theta_d:=\min\left\{b_d(d-1),\ \tfrac1{10}\right\},
  \qquad
  \beta_d:=\frac{\theta_d}{8d}>0 .
  \label{eq:beta-choice}
\end{equation}
By~\eqref{eq:N-coarse-bounds}, $\log_2N(d,m)\leq2dm$, so
$N(d,m)^{-\beta_d}\geq2^{-2d\beta_dm}=2^{-\theta_dm/4}$.  On the other
hand $w_0\geq2^{m-2}$ by Lemma~\ref{lem:input-counts}, so
\[
  N_{RST}(d,m)=m\,w_0\,w^{d-2}
  \geq m\cdot2^{m-2}\left(\tfrac12m2^m\right)^{d-2}
  =2^{-d}m^{d-1}2^{(d-1)m},
\]
whence $\log_2N_{RST}(d,m)\geq(d-1)m-d\geq(d-1)m/2$ for $m\geq4$, and
therefore
$N_{RST}(d,m)^{-b_d}\leq2^{-b_d(d-1)m/2}\leq2^{-\theta_dm/2}$.  Since
$\theta_d\leq\frac1{10}$ we also have
$2C_{d}\cdot 2^{-m/20}\leq2C_{d}\cdot 2^{-\theta_dm/2}$.  Consequently
\begin{equation}
  N_{RST}(d,m)^{-b_d}+2C_{d}\cdot 2^{-m/20}
  \leq2^{-\theta_dm/2}(1+2C_{d})
  \leq2^{-\theta_dm/4}
  \leq N(d,m)^{-\beta_d}
  \label{eq:final-exponent}
\end{equation}
for every large enough $m$.  This
proves~\eqref{eq:G-correlation} for the uncomplemented orientation.

If the formula $Sipser_{d,m}$ has an AND root, complement both exact and modified
targets.  For an AND-root approximating circuit $C$, its complement $1-C$ can be converted using de Morgan's law, to have an OR root, unchanged depth and size, and
\[
  \Prb[C=1-\widehat{Sipser}_{d,m}]=\Prb[1-C=\widehat{Sipser}_{d,m}].
\]
Applying the estimate to the AND-rooted formula $\widehat{Sipser}_{d,m}$ and the OR-rooted circuit $1-C$ (using the same type of couplings described above) proves the normalized
claim.  
\end{proof}

\subsection{Deriving the Main Theorem from the Hardness Lemma.}

Fix $k\geq1$ and set
$d=k+1$.
We construct an infinite language in $\SigmaP{k}{A}$ that is
$\PiP{k}{A}$-immune for almost all oracles $A$. 

Recall that, by Lemma~\ref{lem:wide-spacing}, for every fixed polynomial $q$, there is $j_0\geq 1$ such that
\begin{equation}
   q(t_j)<t_{j+1}
   \qquad\text{for all }j\geq j_0.
   \label{eq:PH-spacing}
\end{equation}

Let $m_\ast(d)$ be sufficiently
large that Lemmas~\ref{lem:R-asymptotic-lemma},~\ref{lem:w0-distance},~\ref{lem:input-counts},~\ref{lem:effective-sipser}, as well as the estimates
\eqref{eq:effective-size}--\eqref{eq:N-coarse-bounds}, hold for every $m\geq m_\ast(d)$.  Choose
\begin{equation}
    J_d :=
    \min\left\{
        j\geq 1 :
        \left\lfloor \frac{t_j}{8d}\right\rfloor
        \geq m_\ast(d)
    \right\}.
\end{equation}
For every $j\geq J_d$, define

\begin{equation}
   m_j=
   \left\lfloor
      \frac{t_j}{8d}
   \right\rfloor
   \label{eq:mj-choice}
\end{equation}
and write
\begin{equation}
   N_j=N(d,m_j).
   \label{eq:Nj-def}
\end{equation}
By~\eqref{eq:effective-size}, there is a constant $\gamma_d>0$ such that for $j\geq J_d$ 
\begin{equation}
   2^{\gamma_dt_j}
   \leq N_j
   \leq 2^{t_j/2}
   <2^{t_j}.
   \label{eq:Nj-range}
\end{equation}

We may therefore identify the $N_j$ leaves of $G_{d,m_j}$ with a
uniformly computable set
$V_j\subseteq\{0,1\}^{t_j}$,
$|V_j|=N_j$.
For example, order the leaves lexicographically by their tuples of
child indices and map the $\ell$th leaf to the $t_j$-bit binary
representation of $\ell$.  Since
$\log N_j=O(t_j)$
and $d$ is fixed, this address is computable in time polynomial in
$t_j$.

For an oracle $A$ and $j\geq J_d$, define
\begin{equation}
   X_j^A=
   G_{d,m_j}\bigl((A(z))_{z\in V_j}\bigr)
   \label{eq:Xj-PH}
\end{equation}
and
\begin{equation}
   L_k^A=
   \left\{
      1^{t_j}:X_j^A=1, j\geq J_d. 
   \right\}
   \label{eq:Lk-PH}
\end{equation}

\begin{lemma}
\label{lem:Lk-in-Sigma}
For every oracle $A$,
\begin{equation}
   L_k^A\in\SigmaP{k}{A}.
   \label{eq:Lk-Sigma}
\end{equation}
\end{lemma}

\begin{proof}
By Lemma~\ref{lem:support-decidable}, the support $T_{d}=\{1^{t_j}:j\geq J_d\}$ 
is decidable in polynomial time.  Thus it remains to evaluate
$G_{d,m_j}$ on input $1^{t_j}$.

The formula has depth $d=k+1$ and an OR root.  The machine
existentially chooses a child of the root, universally chooses a
child of the next layer, and continues alternating in this fashion.
After $d-1=k$ such choices, a bottom gate has been reached.  Its
$m_j=O(t_j)$ leaves are then evaluated deterministically by computing
their addresses in $V_j$, querying $A$, and evaluating the bottom
AND or OR gate.  All choices have length $O(t_j)$.  

One aspect to take care of in the algorithm above is that of invalid (guessed) indices. By Lemma~\ref{lemma-validity} they can be recognized in polynomial time, and are handled as follows: 
\begin{itemize} 
\item[-] on existential levels invalid indices lead to rejection. 
\item[-] on universal levels invalid indices lead to acceptance. 
\end{itemize} 

This is a
$\SigmaP{k}{A}$ computation.
\end{proof}

\begin{lemma} 
Put
\begin{equation}
   \varepsilon_{d,m}=N(d,m)^{-\beta_d}.
   \label{eq:epsilon-dm}
\end{equation}
Then, for $m\geq m_{\ast}(d)$, 
\begin{equation}
   \left|
   \Prb_X[G_{d,m}(X)=1]-\frac12
   \right|
   \leq \varepsilon_{d,m}.
   \label{eq:sipser-balance}
\end{equation}
\end{lemma} 
\begin{proof} 
Indeed, apply Lemma~\ref{lem:effective-sipser} to the constant
zero and constant one functions, using the gadgets in Remark~\ref{rem:constants}
if necessary to obtain the required syntactic depth and top-gate
orientation.
\end{proof} 

\begin{lemma} Define $\varepsilon_j=N_j^{-\beta_d}$. Then: 
\begin{equation}
   \sum_{j=J_d}^{\infty}\varepsilon_j<\infty.
   \label{eq:epsilon-PH-summable}
\end{equation}
\end{lemma}
\begin{proof} 
By~\eqref{eq:Nj-range},
\begin{equation}
   \varepsilon_j
   \leq
   2^{-\beta_d\gamma_dt_j}=2^{-\beta_d\gamma_d\cdot  2^{j!}}
   \label{eq:epsilon-PH-bound}
\end{equation}
The sum of right-hand bounds is finite since, e.g. $2^{j!}\geq j!$. 
\end{proof}

\begin{lemma} For every $k\geq 1$ the language $L_{k}^A$ is almost surely infinite.  
\end{lemma} 
\begin{proof} 
Indeed, the blocks $V_j$
are pairwise disjoint, so the variables $X_j^A$ are independent.
By~\eqref{eq:sipser-balance},
\begin{equation}
   \Prb[X_j^A=1]
   =
   \frac12\pm\varepsilon_j.
   \label{eq:Xj-balance}
\end{equation}
In particular,
\begin{equation}
   \Prb[X_j^A=1]\geq\frac13
   \label{eq:Xj-positive-prob}
\end{equation}
for all sufficiently large $j$.  The ordinary second
Borel--Cantelli lemma therefore gives
\begin{equation}
   \Prb[X_j^A=1\text{ infinitely often}]=1.
   \label{eq:Lk-infinite}
\end{equation}
\end{proof} 

It remains to prove immunity:   

\begin{lemma} 
For each fixed $\Pi^P_k$ clocked oracle machine $M$,
\begin{equation}
   \Prb_A\left[
      L(M^A)\subseteq L_k^A
      \Longrightarrow
      L(M^A)\text{ is finite}
   \right]
   =1.
   \label{eq:machinewise-PH-immunity}
\end{equation}

\end{lemma} 
\begin{proof} 

There is a polynomial $q_M$ such that
on every input $x$ the acceptance predicate of $M^A(x)$, viewed as a
Boolean function of the oracle bits, is represented by a depth-$d$
circuit $C_{M,x}$ satisfying
\begin{align}
   \operatorname{depth}(C_{M,x})&=d=k+1,
   \label{eq:PH-circuit-depth}\\
   \operatorname{top}(C_{M,x})&=\mathrm{AND},
   \label{eq:PH-circuit-top}\\
   \operatorname{size}(C_{M,x})&\leq2^{q_M(|x|)},
   \label{eq:PH-circuit-size}\\
   |z|&\leq q_M(|x|)
   \quad\text{for every oracle variable }A(z)
   \text{ occurring in }C_{M,x}.
   \label{eq:PH-query-length}
\end{align}

Fix the machine \(M\).  After padding with unused quantified bits if necessary,
there are a polynomial \(p\), polynomially bounded strings
\(y_1,\ldots,y_k\), and a deterministic polynomial-time oracle machine \(R\)
such that
\[
  M^A(x)=1
  \quad\Longleftrightarrow\quad
  \forall y_1\,\exists y_2\cdots Q_k y_k\;
  R^A(x,y_1,\ldots,y_k)=1,
\]
where the quantifiers alternate, beginning universally.  The running time,
number of oracle queries, and every query length of \(R\) are at most \(p(|x|)\)
after enlarging \(p\).

For fixed \(x,y_1,\ldots,y_k\), expand the adaptive computation of \(R^A\) into
its binary oracle decision tree.  This tree has depth at most \(p(|x|)\) and at
most \(2^{p(|x|)}\) leaves.  It can be expressed either as a DNF, by taking one
term for every accepting leaf, or as a CNF, by taking one clause for every
rejecting leaf.  Repeated oracle variables may be retained or identified; in
either case all negations occur only on input literals.

Choose the DNF representation when \(Q_k=\exists\), and the CNF representation
when \(Q_k=\forall\).  Its top gate then has the same type as the gate
representing the last quantifier and can be merged with that gate.  The
remaining term or clause layer has the opposite type.  The \(k\) quantifier
layers together with this final term/clause layer therefore form an
alternating circuit of depth exactly \(k+1=d\), with an AND root.  All gates at
a fixed level have the same type.  Constant decision-tree leaves are permitted
by Definition~\ref{def:normal-form}; vacuous quantifier blocks and degenerate formulas can be
padded by unary gates or the constant gadgets of Remark~\ref{rem:constants} without changing
the root type.

The number of assignments to all quantified strings and the total size of all
decision-tree formulas are \(2^{p'(|x|)}\) for a polynomial \(p'\) depending
only on \(M\).  Any remaining layered-normal-form conversion has at most a
fixed-degree polynomial cost and can be absorbed by enlarging \(p'\).
Consequently there is a polynomial \(q_M\) for which
\[
  \operatorname{depth}(C_{M,x})=k+1,
  \qquad \operatorname{top}(C_{M,x})=\mathrm{AND},
  \qquad \operatorname{size}(C_{M,x})\leq2^{q_M(|x|)},
\]
and every oracle variable \(A(z)\) occurring in \(C_{M,x}\) satisfies
\(|z|\leq q_M(|x|)\).

By~\eqref{eq:PH-spacing}, for $j\geq j_0$, $q_M(t_j)<t_{j+1}.$ 
Moreover,~\eqref{eq:Nj-range} implies
$q_M(t_j)<N_j^{\alpha_d}$ 
for all sufficiently large $j$, and hence
\begin{equation}
   2^{q_M(t_j)}
   \leq
   2^{N_j^{\alpha_d}}.
   \label{eq:circuit-fits-RST-exp}
\end{equation}

Define
\begin{align}
   \mathcal F_j
   &=
   \sigma\bigl(A(z):|z|<t_{j+1}\bigr),
   \label{eq:Fj-PH}\\
   \mathcal G_j
   &=
   \sigma\bigl(A(z):|z|<t_{j+1},\ z\notin V_j\bigr).
   \label{eq:Gj-PH}
\end{align}
Then
\begin{equation}
   \mathcal F_{j-1}
   \subseteq
   \mathcal G_j
   \subseteq
   \mathcal F_j \subseteq \mathcal G_{j+1}
   \label{eq:F-G-inclusions-PH}
\end{equation}
The last inclusion follows since $V_{j+1}$ consists of strings of length exactly $t_{j+1}$, whereas $\mathcal{F}_j$ sees only strings of shorter length. 

Conditional on $\mathcal G_j$, the bits $(A(z))_{z\in V_j}$ 
remain mutually independent unbiased bits.
Let
\begin{equation}
   E_j=
   \{M^A(1^{t_j})\text{ accepts}\}.
   \label{eq:Ej-PH}
\end{equation}
By~\eqref{eq:PH-spacing}, $
   E_j\in\mathcal F_j.$ 

After conditioning on $\mathcal G_j$, all oracle variables of
$C_{M,1^{t_j}}$ outside $V_j$ become constants.  We perform this
restriction syntactically, without simplifying away gates.
Consequently the alternating layered structure and the AND root
are preserved.  Thus, on each atom of $\mathcal G_j$, $E_j$ is
represented as a function of the fresh block $V_j$ by a
normal-form depth-$d$, AND-root circuit of size at most
$2^{N_j^{\alpha_d}}$.  Lemma~\ref{lem:effective-sipser}
therefore yields

\begin{equation}
   \Prb[\one_{E_j}=X_j^A\mid\mathcal G_j]
   \leq
   \frac12+\varepsilon_j,
   \qquad
   \varepsilon_j=N_j^{-\beta_d}.
   \label{eq:conditional-RST}
\end{equation}
At the same time,
\begin{equation}
   \left|
   \Prb[X_j^A=1\mid\mathcal G_j]-\frac12
   \right|
   \leq\varepsilon_j.
   \label{eq:conditional-balance}
\end{equation}

For $j\geq J_d$ set 
\begin{align}
   P_j
   &=
   \Prb[E_j\mid\mathcal G_j],
   \label{eq:Pj-PH}\\
   B_j
   &=
   E_j\cap\{X_j^A=0\},
   \label{eq:Bj-PH}\\
   Q_j
   &=
   \Prb[B_j\mid\mathcal G_j].
   \label{eq:Qj-PH}
\end{align}
The event $B_j$ is a false positive: $M$ accepts $1^{t_j}$ although
$1^{t_j}\notin L_k^A$.

Let $J=J(M,d)\geq J_d$ be large enough so that~\eqref{eq:PH-spacing},~\eqref{eq:circuit-fits-RST-exp} and all asymptotic hardness results hold for $j\geq J$.  Applying the rare-or-wrong inequality (Lemma~\ref{lem:rare-or-wrong}) to $\one_{E_j}$ and $X_j^A$, $j\geq J$, implies
\begin{equation}
 Q_j\geq \frac12P_j-\varepsilon_j,
 \qquad \sum_{j\geq J}\varepsilon_j<\infty.
\end{equation}
On the interlaced filtration $\mathcal H_{2j-2}=\mathcal F_{j-1}$, $\mathcal H_{2j-1}=\mathcal G_j$, $\mathcal H_{2j}=\mathcal F_j$, apply L\'evy's lemma to the sequences $D_{2j}=E_j,D_{2j-1}=\varnothing$ and $D'_{2j}=B_j,D'_{2j-1}=\varnothing$.  On a common event of probability one:

\begin{itemize}
    \item if $\sum_{j\geq J}P_j<\infty$, only finitely many $E_j$ occur;
    \item if $\sum_{j\geq J}P_j=\infty$, then $\sum_{j\geq J}Q_j=\infty$, so infinitely many $B_j$ occur.
\end{itemize}

In the first case, $L(M^A)\subseteq L_k^A$ forces $L(M^A)$ to be finite because no word outside the support can be accepted.  In the second case, $B_j$ gives a false positive, so $L(M^A)\nsubseteq L_k^A$.  Therefore
\begin{equation}
 \Prb_A\bigl[L(M^A)\subseteq L_k^A\Longrightarrow L(M^A)\text{ is finite}\bigr]=1.
\end{equation}

\end{proof} 

There are only countably many syntactically clocked $\Pi^P_k$
machines.  Intersecting their probability-one events and using
\eqref{eq:Lk-infinite}, we obtain
\begin{equation}
   \Prb_A\left[
      L_k^A\in\SigmaP{k}{A}
      \text{ and }L_k^A
      \text{ is }\PiP{k}{A}\text{-immune}
   \right]
   =1.
   \label{eq:first-half-theorem1}
\end{equation}

For the second separation we consider the symmetric witness
\begin{equation}
   R_k^A=
   \left\{
      1^{t_j}:X_j^A=0, j\geq J_d. 
   \right\}
   \label{eq:Rk-PH}
\end{equation}
Its target function on $V_j$ is $1-G_{d,m_j}$, an alternating
depth-$d$ formula with an AND root.  Hence
$R_k^A\in\PiP{k}{A}.$
For a $\Sigma^P_k$ oracle machine the corresponding restricted
circuit has an OR root, so the dual part of
Lemma~\ref{lem:effective-sipser} applies.  The entire
rare-or-wrong, filtration, and conditional Borel-Cantelli argument
is unchanged.  Consequently
\begin{equation}
   \Prb_A\left[
      R_k^A\in\PiP{k}{A}
      \text{ and }R_k^A
      \text{ is }\SigmaP{k}{A}\text{-immune}
   \right]
   =1.
   \label{eq:second-half-theorem1}
\end{equation}

Finally, there are only countably many integers $k\geq1$.
Intersecting the probability-one events
\eqref{eq:first-half-theorem1} and
\eqref{eq:second-half-theorem1} over all $k$ proves Theorem~1
simultaneously at every level of the polynomial hierarchy.

\section{Conclusions and Open Problems} 
\label{sec:conc}
Although a single general metatheorem covering all the applications of our general method may be hard to state and prove, we believe that (if you have made it thus far in the paper) you may be convinced that the approach we proposed is reasonably flexible, and can conceivably be employed to upgrade many other separation results under a random oracle to separations with immunity. 

It is also somewhat instructive to scrutinize the limits of our method. Its clearest failure was Case (b) of Observation~\ref{thm:bh}, but there isn't that much to learn from this particular case: there the obstruction was structural, rather than a limitation of our rare-or-wrong approach.  Indeed, every language in the Boolean hierarchy is a finite union of $\DP$ languages (and this remains true relative to every oracle).  Consequently, every infinite language in the Boolean hierarchy contains an infinite $\DP$ sublanguage, making the corresponding immunity statement impossible, independently of the choice of oracle or of the proof technique.  This contrasts with the situation of Theorem~\ref{thm:npconp}, where no analogous structural obstruction is present: there the difficulty is probabilistic, and consists in making sure that the exceptional events in the rare-or-wrong argument have summable probabilities. Thus, while the method cannot turn every random-oracle separation into a separation with immunity, its failures may reflect either genuine structural barriers or insufficient \emph{local} hardness. 

On the other hand the method, as described in this paper, does not currently produce the strongest possible separations. Hemaspaandra and Zimand \cite{32} have studied the concept of \emph{balanced immunity} and have shown that with probability 1 (relative to a random oracle $A$) $NP^A$ contains a set that is $P^A$-balanced-immune. 

The witness sets constructed in this paper are not even biimmune, much less balanced immune. The reason is that they include in the witnessing languages words of type $1^{t_j}$ only. The complement of the set $T=\{1^{t_j}:j\geq 1\}$ is polynomial time decidable and included in the complement of the constructed languages $L^{A}$. 
It remains to be seen whether Hemaspaandra and Zimand's more complicated approach to witness language construction can be combined with the "rare-or-wrong" technique from this paper to obtain the given separations with balanced immunity, when appropriate\footnote{Not all strong separations can extend to separations with biimmunity: for instance an infinite set in $\Sigma_{k}^{P,A}$ cannot be $\Pi_{k}^{P,A}$-biimmune.}.

Another interesting open problem concerns \emph{generic oracle separations}. However, generic separations \emph{with immunity} are a topic with very few results: practically the only one we know is from the same paper of Hemaspaandra and Zimand \cite{32}. 

On the other hand, a recent exciting approach has connected separations with respect to a generic oracle in (the different) \emph{Ellentuck topology} \cite{33} to unconditional separations. It is quite interesting (in this author's opinion) that methods based on Ramsey Theory and Descriptive Set Theory may have implications for separating complexity classes. 
Whether a version of this phenomenon extends to separations \emph{with immunity} is a fascinating (and open) research question. 

\section{Acknowledgments} 

I acknowledge the use of AI assistants (precise names available on request) to help me oversell the results of this paper, for idea exploration,  and for editorial assistance as a "virtual referee"\footnote{In particular, the most significant contribution of the LLMs was to single out the issue of polynomial time computability of the fan-ins in an attempted proof of our Theorem~\ref{thm:main} using the original family of Sipser functions from \cite{1}, and to suggest a way around it by employing an approximation of the top fan-in; the suggested fix was incomplete (it did not realize that the number of variables of the reparameterized formulas was different from those of the original ones, and was missing the coupling lemma, whose idea was supplied by the author) and more complicated: the approximation proposed by the LLM guaranteed that $|\widehat{w_0}-w_0|\leq 1$ and ensured polynomial time computability of $\widehat{w_0}$ by a rather baroque argument; the author realized that the coupling lemma allowed a looser approximation $\widehat{w_0}$, with a simplified complexity argument.}. Said LLMs and their parent companies are \textbf{not} responsible, however, for several aspects of the paper, including (i) the main technical ingredient of this paper, the "rare-or-wrong" approach, which existed in a preliminary form (Lemma~\ref{lem:rare-or-wrong}, and the proof of Theorem~\ref{thm:ph-pspace-immune}) before LLMs were a thing; 
(ii) the results in this paper and the writing of the proofs, for which the author assumes full responsibility;  
(iii) some personal touches in the writing.  

On the other hand, I would like to thank all the \emph{people} whose results motivated this paper\footnote{I first asked myself the question settled by Theorem~\ref{thm:main} in 1995, reading \cite{3} as a graduate student; this is a paper I wish I had written back then.}, especially Pavan Aduri, Jin-Yi Cai, Lance Fortnow, the late Juris Hartmanis, Lane Hemaspaandra, Elvira Mayordomo, Ken Regan, J\"org Rothe, Ben Rossman, the late Alan Selman, Marius Zimand. I have had the good fortune to meet every one of the people listed above, and to interact substantially with (and learn from) many of them.



\begin{thebibliography}{99}

\bibitem{51} 
K.~Alladi and M.~L.~Robinson. 
\emph{Legendre polynomials and irrationality}.
J.~Reine Angewandte Mathematik, 318 (1980), pp. 137--155.


\bibitem{23} S. Arora and B. Barak. \emph{Computational complexity: a modern approach.} Cambridge University Press, 2009.

\bibitem{4}
J.~Aspnes, R.~Beigel, M.~Furst, and S.~Rudich. 
\newblock \emph{The expressive power of voting polynomials},
\newblock Combinatorica \textbf{14} (1994), no.~2, pp. 135--148. 







\bibitem{5}
L.~Babai. 
\newblock \emph{Random oracles separate PSPACE from the polynomial-time
hierarchy},
\newblock Information Processing Letters \textbf{26} (1987),
pp. 51--53.
\newblock DOI:
\href{https://doi.org/10.1016/0020-0190(87)90036-6}
{10.1016/0020-0190(87)90036-6}.









\bibitem{43} 
J.L. Balcázar and U. Schöning. \emph{Bi-immune sets for complexity classes.} Mathematical Systems Theory 18(1), 1985, pp. 1--10.


\bibitem{7}
R. Beigel. \emph{Relativized counting classes: Relations among thresholds, parity, and mods.} Journal of Computer and System Sciences, 1991, 42.1: 76-96.



\bibitem{2} C.H.~Bennett and J.~Gill. \emph{Relative to a Random Oracle A, $P^A\neq NP^A\neq \mbox{co-}NP^A$ with Probability 1}. SIAM Journal on Computing, 10(1), 1981, pp. 96--113. \newblock DOI: \href{https://doi.org/10.1137/0210008}{10.1137/0210008}.





\bibitem{12}
D.~Bruschi,
\newblock \emph{Strong separations of the polynomial hierarchy with
oracles: constructive separations by immune and simple sets},
\newblock Theoretical Computer Science \textbf{102}(2)
(1992), 215--252.
\newblock DOI:
\href{https://doi.org/10.1016/0304-3975(92)90232-5}
{10.1016/0304-3975(92)90232-5}.

\bibitem{6}
J.-Y. Cai. 
\newblock \emph{With probability one, a random oracle separates PSPACE from the
polynomial-time hierarchy},
\newblock Journal of Computer and System Sciences \textbf{38}
(1989), no.~1, pp. 68--85.

\bibitem{29} J.-Y. Cai. \emph{Probability One Separation of the Boolean Hierarchy}, In: Brandenburg, F.J., Vidal-Naquet, G., Wirsing, M. (eds) Proceedings of STACS 87. Lecture Notes in Computer Science, vol 247, 
pp. 148--158. Springer Verlag. \href{https://doi.org/10.1007/BFb0039602}{doi:10.1007/BFb0039602}

\bibitem{37} J-Y. Cai, T. Gundermann, J. Hartmanis, L. Hemachandra, V. Sewelson, K. Wagner and G. Wechsung. \emph{The Boolean Hierarchy I: Structural Properties}. SIAM Journal on Computing, 17(6), 1988, pp. 1232--1252. 



\bibitem{30} J-Y. Cai and L. Hemachandra. \emph{The Boolean hierarchy: hardware over NP.} In: Selman, A.L. (eds) Structure in Complexity Theory. Lecture Notes in Computer Science, vol 223, pp. 105-124.  Springer Verlag. \href{https://doi.org/10.1007/3-540-16486-3$\_$93}{doi:10.1007/3-540-16486-3$\_$93}

\bibitem{44} J-Y. Cai, A. Selman. \emph{Fine separation of average-time complexity classes.} SIAM Journal on Computing 28(4), 1999, pp. 1310--1325.

\bibitem{31}
T. K. Chandra. \emph{The Borel-Cantelli Lemma}, SpringerBriefs in Statistics, Springer, 2012.

\bibitem{13}
R.~Chang, B.~Chor, O.~Goldreich, J.~Hartmanis, J.~H{\aa}stad,
D.~Ranjan, and P.~Rohatgi. 
\newblock \emph{The Random Oracle Hypothesis is false.} 
\newblock Journal of Computer and System Sciences,
\textbf{49}(1) (1994), 24--39.
\newblock DOI:
\href{https://doi.org/10.1016/S0022-0000(05)80084-4}
{10.1016/S0022-0000(05)80084-4}.


\bibitem{33} 
A. Creiner and S. Jackson. \emph{Borel complexity and Ramsey largeness of sets of oracles separating complexity classes.} Mathematical Logic Quarterly 69.3 (2023): 267-286.

\bibitem{50}
R. Durrett. \emph{Probability. Theory and Examples}, 4th Edition, Cambridge University Press, 2010. 

\bibitem{46} 
L. Fortnow, A. Pavan, and A. L. Selman. \emph{Distributionally hard languages.} Theory of Computing Systems 34(3), 2001, pp. 245--261.






\bibitem{24}
R.~Ghosal, Y.~Ishai, A.~Korb, E.~Kushilevitz, P.~Lou,
and A.~Sahai. 
\newblock \emph{Hard languages in $NP\cap coNP$ and NIZK proofs from
unstructured hardness.} 
\newblock in \emph{Proceedings of the 55th Annual ACM Symposium on
Theory of Computing (STOC 2023)}, Association for Computing
Machinery, New York, 2023, pp.~1243--1256.
\newblock DOI:
\href{https://doi.org/10.1145/3564246.3585119}
{10.1145/3564246.3585119}.
\newblock Full version: Cryptology ePrint Archive,
Report 2023/1972.



\bibitem{25}
O.~Goldreich,
\newblock \emph{Foundations of Cryptography, Volume 1: Basic Tools},
\newblock Cambridge University Press, Cambridge, 2001.
\newblock DOI:
\href{https://doi.org/10.1017/CBO9780511546891}
{10.1017/CBO9780511546891}.

\bibitem{26}
O.~Goldreich,
\newblock \emph{On one-way permutations (a clarification)}. 
\newblock technical note,
Weizmann Institute of Science.
\newblock Available from
\url{https://www.wisdom.weizmann.ac.il/~oded/owp.html}.

\bibitem{27}
O.~Goldreich, L.~A. Levin, and N.~Nisan,
\newblock \emph{On constructing 1-1 one-way functions}. 
\newblock in O.~Goldreich (ed.), Studies in Complexity and Cryptography:
Miscellanea on the Interplay between Randomness and Computation,
Lecture Notes in Computer Science, vol.~6650,
Springer, 2011, pp.~13--25.
\newblock DOI:
\href{https://doi.org/10.1007/978-3-642-22670-0_3}
{10.1007/978-3-642-22670-0\_3}.

\bibitem{15}
S.~Goldwasser, S.~Micali, and C.~Rackoff,
\newblock \emph{The knowledge complexity of interactive proof systems}. 
\newblock SIAM Journal on Computing, \textbf{18}(1)
(1989), 186--208.
\newblock DOI:
\href{https://doi.org/10.1137/0218012}
{10.1137/0218012}.


\bibitem{16}
S.~Goldwasser and M.~Sipser. 
\newblock \emph{Private coins versus public coins in interactive proof
systems}
\newblock in Proceedings of the 18th Annual ACM Symposium on
Theory of Computing (STOC 1986), pp.~59--68.
\newblock DOI:
\href{https://doi.org/10.1145/12130.12137}
{10.1145/12130.12137}.




\bibitem{61} 
J. Håstad. \emph{Computational Limitations for Small-Depth Circuits}. M.I.T. Press, 1987. 

\bibitem{34} 
J. Håstad. \emph{On the correlation of parity and small-depth circuits.} SIAM Journal on Computing 43.5 (2014): 1699-1708.

\bibitem{39} J. Håstad, B. Rossman, R. Servedio and L.Y. Tan. \emph{An average-case depth hierarchy theorem for boolean circuits}. Journal of the A.C.M., 64(5), 2017, pp. 1--27.

\bibitem{32}
L. Hemaspaandra and M. Zimand. 
\emph{Strong self-reducibility precludes strong immunity.} Mathematical systems theory 29.5 (1996): 535-548.


\bibitem{48} 
J. Hitchcock, A. Pavan, and N. V. Vinodchandran. \emph{Partial bi-immunity, scaled dimension, and NP-completeness.} Theory of Computing Systems 42(2), 2008, pp. 131--142.

\bibitem{45} 
S. Homer and W. Maass. \emph{Oracle-dependent properties of the lattice of NP sets.} Theoretical Computer Science 24.3 (1983), pp. 279--289.

\bibitem{65} 
G. Istrate. \emph{Game-theoretic Models of Moral and Other-Regarding Agents (extended abstract)}. Proceedings of TARK 2021, 18th Conference on Theoretical Aspects of Rationality and Knowledge, Electronic Proceedings of Theoretical Computer Science, vol. 335. pp. 217--227. 

\bibitem{42} C. G. Jockusch, Jr. \emph{The degrees of bi-immune sets.} Mathematical Logic Quarterly 15.7‐12 (1969), pp. 135--140.

\bibitem{36} 
J. Katz and Y. Lindell. \emph{Introduction to Modern Cryptography, 3rd ed.} CRC Press, 2020. 

\bibitem{38} 
Ker-I Ko. \emph{A note on separating the relativized polynomial time hierarchy by immune sets.} RAIRO. Informatique théorique et applications 24.3 (1990): pp. 229--240.





\bibitem{21}
C.~Lund, L.~Fortnow, H.~Karloff, and N.~Nisan.
\newblock \emph{Algebraic methods for interactive proof systems},
\newblock Journal of the A.C.M. \textbf{39}(4)
(1992), 859--868.
\newblock DOI:
\href{https://doi.org/10.1145/146585.146605}
{10.1145/146585.146605}.

\bibitem{49}
E. Mayordomo. \emph{Almost every set in exponential time is P-bi-immune.} Theoretical Computer Science 136(2), 1994, pp. 487--506.

\bibitem{54} 
A. Pavan and A.L. Selman. \emph{Complete distributional problems, hard languages, and resource-bounded measure.} Theoretical Computer Science 234.1-2 (2000): 273-286.

\bibitem{47} 
A. Pavan and A.L. Selman. \emph{Bi-immunity separates strong NP-completeness notions.} Information and Computation 188(1), 2004, pp. 116--126.

\bibitem{41} E. L. Post. \emph{Recursively enumerable sets of positive integers and their decision problems.} Bulletin of the A.M.S. 50(5), 1944, pp. 284--316. 

\bibitem{60} K. W. Regan and J. Royer. \emph{On Closure Properties of Bounded Two-Sided Error Complexity Classes}. Mathematical Systems Theory \textbf{28}, 1995, pp. 229--243. 


\bibitem{1} B. Rossman, R. A. Servedio, and L.-Y. Tan,
\emph{An Average-Case Depth Hierarchy Theorem for Boolean Circuits}. In Proceedings of the 56th Annual Symposium on Foundations of Computer Science (FOCS'15), 2015, pp. 1030-1048. Full version: \url{https://arxiv.org/abs/1504.03398v1}.



\bibitem{17} B. Rossman, R. Servedio, and L.-Y. Tan. \emph{The Polynomial Hierarchy, Random Oracles, and Boolean Circuits.} ACM SIGACT News 46(4), 2015, pp. 50--68. DOI: 10.1145/2852040.2852052.

\bibitem{40}
Jörg Rothe. \emph{Immunity and Simplicity for Exact Counting and Other Counting Classes.} RAIRO 33(2) (1999), pp. 159--176.



\bibitem{22}
A.~Shamir,
\newblock \emph{IP = PSPACE},
\newblock Journal of the A.C.M.\textbf{39}(4)
(1992), 869--877.
\newblock DOI:
\href{https://doi.org/10.1145/146585.146609}
{10.1145/146585.146609}.

\bibitem{28}
G.~Tardos,
\newblock \emph{Query complexity, or why is it difficult to separate
$NP^A\cap coNP^A$ from $P^A$ by random oracles $A$?}.
\newblock \emph{Combinatorica} \textbf{9}(4)
(1989), 385--392.
\newblock DOI:
\href{https://doi.org/10.1007/BF02125350}
{10.1007/BF02125350}.

\bibitem{3} N. K. Vereshchagin. \emph{Relationships between NP-sets, co-NP-sets, and P-sets relative to random oracles.} In Proceedings of the Eighth Annual Structure in Complexity Theory Conference. IEEE, 1993, pp. 132-138. \href{https://doi.org/10.1109/SCT.1993.336533}{doi:10.1109/SCT.1993.336533}

\bibitem{63} N.K. Vereshchagin, \emph{NP-sets are Co-NP-immune relative to a random oracle.} Proceedings Third Israel Symposium on the Theory of Computing and Systems, pp. 40--45, 1995.  \href{https://doi.org/10.1109/ISTCS.1995.377048}{doi:10.1109/ISTCS.1995.377048}.
\end{thebibliography}
\end{document}